\documentclass[12pt,times]{article}

\usepackage{graphicx}
\usepackage{epsfig}
\usepackage{amssymb}
\usepackage[left,pagewise,mathlines]{lineno}
\newtheorem{theorem}{Theorem}

\newtheorem{lemma}[theorem]{Lemma}
\newtheorem{corollary}[theorem]{Corollary}

\newenvironment{proof}{
\par
\noindent {\bf Proof.}\rm}%
{\mbox{}\hfill\rule{0.5em}{0.809em}\par}

\title{\bf \normalsize
\Large Target set selection problem for honeycomb networks}

\author{\small Chun-Ying~Chiang\thanks{Partially
supported by National Science Council under
grant NSC100-2115-M-008-007-MY2.},
 Liang-Hao~Huang\thanks{Partially supported by
 National Science Council under grant NSC100-2811-M-008-052.},
 Hong-Gwa Yeh\thanks{Partially supported by National Science Council under
grant NSC100-2115-M-008-007-MY2.}
 \thanks{Corresponding author (hgyeh@math.ncu.edu.tw)}\\
{\footnotesize \em Department of Mathematics, National Central
University, Jhongli City, Taiwan}}

\date{}

\begin{document}
\maketitle \baselineskip=17pt
\begin{abstract}
 Let $G$ be a graph with
 a threshold function $\theta:V(G)\rightarrow \mathbb{N}$ such that
 $1\leq \theta(v)\leq d_G(v)$ for every vertex $v$ of $G$,
 where $d_G(v)$ is the degree of $v$ in
 $G$. Suppose we are given a target set $S\subseteq V(G)$.
 The paper considers the following repetitive process on $G$.
 At time step $0$ the vertices of $S$ are colored black and the
 other vertices are colored white. After that, at each time step $t>0$,
 the colors of white vertices (if any) are updated according
  to the following rule.
   All white vertices $v$ that have at least
   $\theta(v)$ black neighbors at the time step $t-1$
   are colored black, and the colors of the other vertices do not change.
 The process runs until
 no more white vertices
 can update colors from white to black.
 The following optimization problem is called
          T{\scriptsize ARGET} S{\scriptsize ET}
          S{\scriptsize ELECTION}:
          Finding a target set $S$ of smallest possible size
          such that all vertices in $G$  are black at the
          end of the  process. Such an $S$ is called
          an {\em optimal target set} for $G$
          under the threshold function $\theta$.
 We are interested in finding an optimal target set for
 the well-known class of honeycomb networks under an important threshold function
 called {\em strict majority threshold}, where
 $\theta(v)=\lceil
   (d_G(v)+1)/2\rceil$ for each vertex $v$ in $G$.
% Honeycomb networks and strict majority threshold have been
% extensively studied, both respectively and together,
% for their theoretical interest and for
% their practical applications in many fields
% such as the spread of influence in social networks,
% overcoming failure in distributed
% computing and
% parallel systems.
  In a graph $G$, a {\em feedback vertex set} is a subset $S\subseteq V(G)$
  such that the subgraph induced by $V(G)\setminus S$ is cycle-free.
   In this paper
   we give exact value on the size
   of the optimal target set for various kinds of honeycomb networks
   under strict majority threshold,
   and as a by-product
   we also provide a minimum feedback vertex set for
   different kinds regular graphs in the class of honeycomb networks.
   \medskip

\end{abstract}
\parskip 4pt

  %%%%%%%%%%%%%%%%%%%%%%
  %
  %
  %   Section 1:
  %
  %
  %%%%%%%%%%%%%%%%%%%%%%
  %%%%%%%%%%%%%%%%%%%%%%%%%%%%%%%%%%%%%%%%%%%%%%%%
   \section{Introduction}
   \label{intro}
  %%%%%%%%%%%%%%%%%%%%%%%%%%%%%%%%%%%%%%%%%%%%%%%%
  %Internet enables you to get connected to your friends
  Most computer users connect with their friends through email, social networks and chatting applications.
  Recently some popular social networks, such as Facebook, YouTube, Twitter and blogs, have become one of the most important ways for companies to market themselves.
  A lot of companies use viral-marketing techniques to advertise their products.
  Viral marketing is used to quickly spread the word about their products and brands.
  Individual decisions are influenced by others.
  In viral marketing, company tries to target some small number of people to ``seed" their product; advertising spread from one person to another by people talking about it - a kind of snowball effect; and then
  the advertising reaches nearly every potential customer.

  %
  %
  %   Scenario
  %
  %
 Consider the following hypothetical scenario
 as a motivating example.
 A company wish to market a new product.
 The company has at hand a description of the social network $G$
 formed among a sample of potential customers, where the
 vertices represent customers and edges connect people to
 their friends.
 The company wants to target key potential customers $S$ of the social network and persuade them into adopting the new product
 by handing out free samples.
 We assume that individuals in $S$ will be convinced to adopt the new product after they receive a free sample, and
 the friends of customers in $S$ would be
 persuaded into buying the new product, which in turn will
 recommend the product to other friends.
 The company hopes that by word-of-mouth effects,
 convinced vertices in $S$
 can trigger a cascade of further adoptions, and finally
 all potential customers will be persuaded to buy the product.
 But now how to find a good set of potential customers $S$
 to target? To study this problem, in the following we formally define it.

  %  Graph
 A graph $G$ consists of a set $V(G)$ of {\em vertices}
 together with a set $E(G)$ of unordered pairs of
 vertices called {\em edges}.
 We use $uv$ for an edge $\{u,v\}$.
 The {\em degree} of a vertex $v\in V(G)$ is
 the number of vertices adjacent to $v$
 and is denoted by $d_G(v)$
 (the subscript $G$ will be dropped if no confusion
 can arise).
  %
  % Social network
  %
   A person-to-person recommendation social network is
   usually modeled by a graph $G$ together with a
   threshold function $\theta : V(G)\rightarrow \mathbb{N}$
   such that  $1\leq \theta(v)\leq d_G(v)$ for each vertex $v$ in $G$, and such a social network is denoted by $(G,\theta)$.
  For the sake of convenience if $\theta(v)=k$ for all vertices $v$ in $G$, then
  $(G,\theta)$ shall be abbreviated to $(G,k)$.
  In marketing setting, the threshold of a vertex (customer) $v$
  represents his/her latent tendency of buying
  the new product when his/her neighbors (friends) do (see \cite{watts2007HBR}).
   There are two types of important and well-studied thresholds
   on a graph $G$
   called {\em majority threshold}  and {\em strict majority
   threshold} (see
   \cite{Adams-Hexagonal-Grids2011,Chang2010,Khoshkhah-Zaker2012,
   Kilgour2010,Linial-Peleg1993,Peleg1998,
   Peleg2002} and references therein), which will be denoted by
   $\theta_{\geq}$
   and $\theta_{>}$ respectively throughout this paper.
    In a majority threshold
   we have $\theta_{\geq}(v)=\lceil
   d(v)/2\rceil$ for every vertex $v$ of $G$,
   while in a strict majority threshold we have
   $\theta_{>}(v)=\lceil
   (d(v)+1)/2\rceil$ for every vertex $v$ of $G$.

  Given a vertex subset $S$ of a connected social network $(G,\theta)$.
  Consider the following repetitive process played on $(G,\theta)$
  called {\em activation process on $(G,\theta)$ starting from  $S$}.
  At round $0$ (the beginning of the game),
  the vertices of $S$ are colored black and
  the other vertices are colored white.
  After that, at each round $t>0$,
  the colors of white vertices (if any) are updated according
  to the following rule:
   \begin{description}
  \item[Parallel updating rule:]
   All white vertices $v$ that have at least
   $\theta(v)$ black neighbors at the previous round $t-1$
   are colored black. The colors of the other vertices do not change.
   \end{description}

   \noindent
   The process runs until
   %either all vertices are black or
   no more white vertices
   can update colors from white to black.
   The set $S$ is called a {\em target set} for $(G,\theta)$.
   We denote by $[S]_\theta^G$ the set of vertices
   that are black at the end of the process.
   If $F\subseteq [S]_\theta^G$, then we say that the target set
   $S$ {\em influences} $F$ on $(G,\theta)$.
   We are interested in the following optimization problem:
     \begin{description}
          \item[T{\scriptsize ARGET} S{\scriptsize ET}
          S{\scriptsize ELECTION}:]
          Finding a target set $S$ of smallest possible size
          such that all vertices in $(G,\theta)$  are black at the
          end of the activation process starting from $S$.
          Such an $S$ is called
          an {\em optimal target set} for $(G,\theta)$
          and its size is denoted by {\rm min-seed}$(G,\theta)$.
   \end{description}

  %
  %
  %  History
  %
  %
  The theoretical investigations of certain kinds of
  target set selection problem
  were initiated by
   Kempe, Kleinberg, and Tardos in \cite{kkt2003,kkt2005},
   where they mainly consider
   probabilistic thresholds such that all thresholds are
   drawn randomly from a given distribution.
   They focused on the maximization
   problem - find a target set of a given size $k\in \mathbb{N}$ to maximize the expected number of black vertices at the end of the activation process.

   Many authors have investigated
   target set selection problem with
   different types of thresholds and network structures
   in various settings and under a variety of assumptions.
   In a dynamic monopoly setting, Peleg \cite{Peleg2002}
   proved that it is NP-hard to
   compute the optimal target set for majority thresholds.
   In constant threshold setting,
   Dreyer and Roberts \cite{Dreyer+Roberts} showed that
   it is NP-hard to compute the min-seed$(G,k)$ for any
   $k\geq 3$, and Chen \cite{chen} showed that
   the target set selection problem is NP-hard when
   the thresholds are at most $2$,
   even for bounded bipartite graphs.
   %it is also NP-hard to compute min-seed$(G,2)$.
   In fact, this problem is not only NP-hard, it is
   also extremely hard to solve approximately.
   Chen \cite{chen} proved that
   min-seed$(G,\theta_\geq)$
   cannot be approximated within the ratio
   $O(2^{\log^{1-\epsilon}n})$ for any fixed constant $\epsilon>0$, unless $NP\subseteq DTIME(n^{polylog(n)})$, where $n=|V(G)|$.

  We now turn to determine the exact value of min-seed$(G,\theta)$ for certain families of graphs $G$ under specific threshold functions $\theta$.
  Related results can be found in
  \cite{Ackerman2010,Adams-Hexagonal-Grids2011,Ben-zwi2011,Berger2001,chen,Yeh2011,
  Dreyer+Roberts,Flocchini2001,Flocchini2003,Flocchini2004,Flocchini2009,Luccio1998,
  Luccio1999,Peleg1998,Peleg2002,SIAM2005decycling-two-cycles,zaker2012}, where min-seed$(G,\theta)$ has
  been investigated  for different types of network structure $G$
  such as
  bounded treewidth graphs,
  hexagonal grids,
  trees,
  cycle permutation graphs,
  generalized Petersen graphs,
  block-cactus graphs,
  chordal graphs,
  Hamming graphs,
  chordal rings,
  tori,
  meshes,
  butterflies,
  Cartesian products of two cycles.

   Majority threshold model has many applications in
   distributed computing
   such as maintaining data consistency in a distributed system,
   fault-local mending in distributed network and
   overcoming failure in distributed computing
    \cite{Kutten-Peleg1995,Linial-Peleg1993,Peleg1998,Peleg2002}.
   On the other hand,
   honeycomb networks have been suggested as an attractive
   architecture for interconnected networks
   which have been widely investigated in
   parallel and distributed applications
   (see \cite{Alspach-Dean2009,Stojmenovic1997} and references therein).
   %the area of distributed systems design \cite{Stojmenovic1997}.
   %
   In this paper, we study
   target set selection problem under strict majority thresholds
   on different kinds of honeycomb networks such as
   honeycomb mesh ${\rm HM}_t$,
   honeycomb torus ${\rm HT}_t$,
   honeycomb rectangular torus ${\rm HReT}(m,n)$,
   honeycomb rhombic torus ${\rm HRoT}(m,n)$,
   generalized honeycomb rectangular torus ${\rm GHT}(m,n)$,
   planar hexagonal grid ${\rm PHG}(m,n)$,
   cylindrical hexagonal grid ${\rm CHG}(m,n)$, and
   toroidal hexagonal grid ${\rm THG}(m,n)$
   (all terms will be defined in later sections).

   In Section \ref{HM} we determine the exact value of min-seed$(G,\theta_>)$ for any honeycomb mesh $G$.
   In Section \ref{GHRT}, by computing the optimal target set
   for a generalized honeycomb rectangular torus, we
   determine the exact values of
   min-seed$(G,\theta_>)$ when $G$ is a
   honeycomb torus or a
   honeycomb rectangular torus or a
   honeycomb rhombic torus.
   Finally, in Section \ref{sec-hexagonal-grids},
   we compute min-seed$(G,\theta_>)$ for planar, cylindrical, and toroidal hexagonal grids $G$, where
   $\theta_>$ denote the strict majority threshold of $G$.
   Our results in Section \ref{sec-hexagonal-grids}
   are summarized in Table 1.

  %
  %
  %  Table 1
  %
  %
  \begin{center}
  \begin{tabular}{ll|ll}
  \hline
  Structure of $G$   &&&  Result \\
  \hline
  Planar             &&&  {\rm min-seed}$(G,\theta_>)=
                   \lceil{mn+2m+n\over 4}\rceil-1$ \\

  Cylindrical        &&&  {\rm min-seed}$(G,\theta_>)=
                   \lceil{mn+2m\over 4}\rceil$ \\

  Toroidal          &&&  {\rm min-seed}$(G,\theta_>)=
                   \lceil{mn+2\over 4}\rceil$ \\
  \hline
  \end{tabular}
  \end{center}
  \begin{description}
    \item[Table 1.]  Summary of results on {\rm min-seed}$(G,\theta_>)$
     where $G$ is an $m$ by $n$ hexagonal
    grid with $m\geq 2$, $n\geq 4$, and $n$ even.
  \end{description}

   A subset $S$ of $V(G)$ is a {\em feedback vertex set}
   (or a {\em decycling set}) of a graph $G$ if the subgraph of $G$
   induced by the vertices in $V(G)\setminus S$ is acyclic
   (see \cite{feedback-set2000,SIAM2005decycling-two-cycles} and
   references therein).
   The size of a minimum feedback vertex set in
   a graph $G$ is called the {\em decycling number} of $G$
   and is denoted by $\nabla(G)$ (adapted from \cite{SIAM2005decycling-two-cycles}).
   In \cite{Adams-Hexagonal-Grids2011},
   by using feedback vertex sets for graphs,
   Adams, Troxell and Zinnen
   show lower and upper bounds for
   min-seed$(G,\theta_\geq)$ when $G$ is one of the graphs
   planar, cylindrical, and toroidal hexagonal grids.
   We summarize their results in Table 2, where
   $\theta_\geq$ denote the majority threshold of $G$.
 Since toroidal hexagonal grids  are $3$-regular,
 it can readily be seen that  if $G$
 is a toroidal hexagonal grid, then
 {\rm min-seed}$(G,\theta_{\geq})=\mbox{
 min-seed}(G,\theta_{>})$.
 Thus
 our result for toroidal hexagonal grids (see Table 1)
 closes the gap in the corresponding result of Table 2.

  %
  %
  %  Table 2
  %
  %
  \begin{center}
  \begin{tabular}{ll|ll}
  \hline
  Structure of $G$   &&&  Result \\
  \hline
  Planar             &&&  {\rm min-seed}$(G,\theta_{\geq})=
                   \lceil{(n-2)(m-1)\over 4}\rceil$ \\

  Cylindrical        &&&  {\rm min-seed}$(G,\theta_{\geq})
                  \in \{\lceil{(n-2)m+2 \over 4}\rceil,
                 \lceil{(n-2)m+2 \over 4}\rceil+1\}$ \\

  Toroidal          &&&  {\rm min-seed}$(G,\theta_{\geq})
                  \in \{\lceil{mn+2 \over 4}\rceil,
                  \lceil{mn+2 \over 4}\rceil+1\}$ \\
  \hline
  \end{tabular}
  \end{center}
  \begin{description}
    \item[Table 2.]  Summary of results on {\rm min-seed}$(G,\theta_{\geq})$
    proved in \cite{Adams-Hexagonal-Grids2011} where $G$ is an $m$ by $n$ hexagonal
    grid with $m\geq 2$, $n\geq 4$, and $n$ even.
  \end{description}

 In \cite{Dreyer+Roberts}, Dreyer and Roberts show that, for a vertex
 subset $S$ of a $(k+1)$-regular graph $G$,
 the target set
 $S$ can influence all vertices of $V(G)\setminus S$
 in the social network $(G,k)$
 if and only if $S$ is a feedback vertex set of $G$.
 In \cite{Adams-Hexagonal-Grids2011},
 the authors further show that if $G$ is a graph
 with minimum degree at least $2$,
 maximum degree at most $3$ and $S\subseteq V(G)$,
 then $S$ can influence all vertices of $V(G)\setminus S$ in the social network $(G,\theta_\geq)$ if and only if
 $S$ is a feedback vertex set of $G$.

 Finding a minimum feedback vertex set of a graph is quite
 difficult, and has been proved to be NP-complete in general
 \cite{karp1972}.
 However, in this paper by using the above facts, we are
 able to provide a minimum feedback vertex set in
 %exact values for the minimum size of a
 %feedback vertex set $\nabla(G)$ when $G$
   honeycomb torus networks,
   honeycomb rectangular torus networks,
   honeycomb rhombic torus networks,
   generalized honeycomb rectangular torus networks,
   and
   toroidal hexagonal grid networks.
   %In this paper, by the above facts, our results
   %on optimal target sets
   %for $3$-regular networks
   %can be restated in terms of minimum feedback vertex sets.
   %That is, in this paper, we also determine the
   %decycling number of

  %%%%%%%%%%%%%%%%%%%%%%
  %
  %
  %   Section 2:
  %
  %
  %%%%%%%%%%%%%%%%%%%%%%
  %%%%%%%%%%%%%%%%%%%%%%%%%%%%%%%%%%%%%%%%%%%%%%%%
   \section{Notations and preliminary results}
   \label{pre}
  %%%%%%%%%%%%%%%%%%%%%%%%%%%%%%%%%%%%%%%%%%%%%%%%
  %
  % Sequential influence diffusion process
  %
  %
  In this section, we introduce
  the necessary notations, definitions and preliminary results
  which will be used through the paper.
   For a set $S\subseteq V(G)$, the {\em
  subgraph of $G$ induced by $S$} is the graph  with vertex set
  $S$ and edge set $\{ uv \in E(G): u,v \in S\}$ and is
  denoted by $G[S]$. Denote by $G\setminus S$ the
  subgraph of $G$ induced by $V(G)\setminus S$.
  %and, for convenience,
  %we write $G-v$ for $G-\{v\}$.
   In order to study the optimal target sets for $(G,\theta)$
   we introduce a sequential version of
   activation process on $(G,\theta)$,
   called {\em sequential activation process}
   in which at each round $t>0$ one employs the following sequential updating rule
   instead of the parallel updating rule:

   \begin{description}
   \item[Sequential updating rule:]
          Exactly one of
          white vertices that have at least $\theta(v)$
          black neighbors at the previous round $t-1$
          is colored black. The colors of the other vertices do not change.
   \end{description}

  %
  %  convinced sequence
  %
   Given a target set $S$ for $(G,\theta)$,
   consider a
   sequential activation process
   starting from $S$.
   In this process, if $v_1,v_2,\ldots,v_r$ is the order that vertices
   in $[S]^G_\theta\setminus S$ become black,
   then $[v_1,v_2,\ldots,v_r]$
   is called the {\em convinced sequence}
   of $S$ on $(G,\theta)$.
   We define an operation $\sqcup$
   on convinced subsequences
   $\alpha=[v_1,v_2,\ldots,v_r]$
   and $\beta =[u_1,u_2,\ldots,u_s]$ as follows:
   $\alpha\sqcup \beta=[v_1,v_2,\ldots,v_r,u_1,u_2,\ldots,u_s]$.
   For a list of convinced subsequences $\{\alpha_{i,j}\}_{1\leq i\leq k,
   1\leq j\leq \ell}$, the sequences $\sqcup_{i=1}^k\alpha_{i,j}$ and
   $\sqcup_{j=1}^\ell\sqcup_{i=1}^k \alpha_{i,j}$ are
   defined to be
   $$\sqcup_{i=1}^k\alpha_{i,j}=\alpha_{1,j}\sqcup\alpha_{2,j}\sqcup
   \cdots\sqcup\alpha_{k,j} \mbox{ and }
   \sqcup_{j=1}^\ell\sqcup_{i=1}^k \alpha_{i,j}=
   \sqcup_{j=1}^\ell(\sqcup_{i=1}^k \alpha_{i,j}).$$

    %
    %  Vertex-ordering
    %
    A {\em vertex-ordering} $\pi$ of a graph $G$ having $n$ vertices
    is a numbering $(v_1,v_2,\ldots,v_n)$ of $V(G)$.
    For an edge $v_iv_j$ with $i<j$, $v_j$ is a {\em successor} of
    $v_i$,
    and $v_i$ is a {\em predecessor} of $v_j$.
    The number of predecessors and successors of a vertex $v_k$
    is denoted by ${\rm pred}_\pi(v_k)$ and ${\rm succ}_\pi(v_k)$,
    respectively. We may omit the subscript $\pi$ if the ordering
    is clear.
    %
    % Lemma: Sequential version = Parallel version
    %
   The proof of the following  lemma is straightforward and
   so is omitted. This lemma will be used frequently in
   the sequel, sometimes without explicit reference to it.
   \begin{lemma}\label{seq=para}
   Let $(G,\theta)$ be a connected graph $G$ with
   thresholds $\theta$ on the vertices of $G$.
   $(a)$ An optimal target set for $(G,\theta)$ under the sequential updating
   rule is also an optimal target set for $(G,\theta)$
   under the parallel updating
   rule, and vice versa. $(b)$ Finding an optimal target set $S$ for
   $(G,\theta)$ is equivalent to that of finding a set $S\subseteq V(G)$
   of minimum possible cardinality such that $G\setminus S$
   has a vertex-ordering $(v_1,v_2,\ldots,v_{|V(G\setminus S)|})$
   with the following property: for each $1\leq i\leq |V(G\setminus
   S)|$, $v_i$ is adjacent to at least $\theta(v_i)$ vertices
   in the set $S\cup \{v_j: j\leq i-1\}$.
   \end{lemma}

  We use similar ideas %a similar idea
  as in the proof of Theorem 1 of
  \cite{zaker2012} to show the results in Lemma
  \ref{lowerbound-for-target-set}
  which generalizes Lemma 4 of \cite{Yeh2011-GPG}.

  %
  % Lower bound
  %
  \begin{lemma}
  \label{lowerbound-for-target-set}
  Let $(G,\theta)$ be a connected graph $G$ with
  thresholds $\theta$ on $V(G)$ and let $\Delta$
  be the maximum degree of $G$.
  Let $n=|V(G)|$, $m=|E(G)|$,
  $\delta=\max\{d_G(v)-\theta(v): v\in V(G)\}$,
  $\theta_V=\sum_{v\in V(G)}\theta(v)$,
  $\theta_{\max}=\max\{\theta(v):v\in V(G)\}$ and
  $\theta_{\min}=\min\{\theta(v):v\in V(G)\}$.
  Then the following quantity $\Lambda$ is a lower bound for \mbox{{\rm
  min-seed}$(G,\theta)$}:
  \begin{linenomath}
  $$\Lambda=\max\left\{
     {m-(n-1)\delta \over \Delta-\delta},
     {\theta_V-m\over \theta_{\max}},
     {n\theta_{\min}-m\over \theta_{\min}},
     {\theta_V-(n-1)\delta\over \Delta-\delta+\theta_{\max}},
     {n\theta_{\min}-(n-1)\delta\over \Delta-\delta+\theta_{\min}}
                \right\}.
  $$
  \end{linenomath}
  \end{lemma}
  %
  % Proof
  %
  \begin{proof}
  Let $S$ be an optimal target set for $(G,\theta)$ and let
  $V=V(G)$, $s=|S|$ and $\ell= n-s$.
  For any two subsets $A,B\subseteq V$,
  let $E(A,B)$ denote the number of edges between $A$ and $B$.
  Since $S$ (sequentially) influences all vertices of
  $(G,\theta)$, $G\setminus S$
   has a vertex-ordering $\pi=(v_1,v_2,\ldots,v_{\ell})$
   with the following property: for each $1\leq i\leq \ell$,
   $v_i$ is adjacent to at least $\theta(v_i)$ vertices
   in the set $S\cup \{v_j: j\leq i-1\}$, and hence
   ${\rm succ}_\pi (v_i)\leq d_G(v_i)-\theta(v_i)$.
   It follows that
   $|E(G\setminus S)|=\sum_{i=1}^{\ell -1}{\rm succ}_\pi (v_i)\leq
   \sum_{i=1}^{\ell -1}(d_G(v_i)-\theta(v_i))\leq (n-s-1)\delta$.
   Note that
  if $e$ is an edge in $E(G)$ but not in $E(G\setminus S)$,
  then $e$ has an end in $S$.
  This leads to $|E(G\setminus S)|\geq m-s\Delta$.
  Thus $m-s\Delta\leq (n-s-1)\delta$, and hence
  $s\geq {m-(n-1)\delta\over \Delta-\delta}$.
  To prove the remaining part of the lemma,
  we see that
  \begin{linenomath}
  \begin{eqnarray*}
    \min\left\{m,s\Delta+(n-s-1)\delta\right\}
    &\geq & \min\left\{m,s\Delta + \sum_{i=1}^{\ell -1}{\rm succ}_\pi (v_i)\right\} \\
    &\geq & \min\left\{m,E(S,\{v_j:1\leq j\leq \ell\})
                    + \sum_{i=1}^{\ell -1}{\rm succ}_\pi (v_i)\right\} \\
     &=& \sum_{i=1}^{\ell } E(S\cup\{v_j:j\leq i-1\},\{v_i\})\\
     &\geq & \sum_{i=1}^{\ell }\theta(v_i)
     \geq
     \max\left\{\theta_V-s\theta_{\max},(n-s)\theta_{\min}\right\}.
  \end{eqnarray*}
  \end{linenomath}
  This implies the following four inequalities:
   $m \geq \theta_V-s\theta_{\max}$,
   $m \geq (n-s)\theta_{\min}$,
   $s\Delta+(n-s-1)\delta \geq \theta_V-s\theta_{\max}$, and
   $s\Delta+(n-s-1)\delta \geq (n-s)\theta_{\min}$.
   After simple algebraic manipulations, we obtain
  $s\geq
     {\theta_V-m\over \theta_{\max}}$,
  $s\geq
     {n\theta_{\min}-m\over \theta_{\min}}$,
  $s\geq
     {\theta_V-(n-1)\delta\over \Delta-\delta+\theta_{\max}}$,
  $s\geq
     {n\theta_{\min}-(n-1)\delta\over \Delta-\delta+\theta_{\min}}$,
     respectively,
     which complete the proof of the lemma.
  \end{proof}
  We remark that the result
  min-seed$(G,\theta)\geq {n\theta_{\min}-m\over \theta_{\min}}$
  shown
  in Lemma \ref{lowerbound-for-target-set} has already
  appeared in Corollary 2 of
  \cite{zaker2012}.

   %%%%%%%%%%%%%%%%%%%%%%
   %
   %
   %   Section 3:
   %
   %
   %%%%%%%%%%%%%%%%%%%%%%
   \section{Honeycomb mesh}
   \label{HM}
   In this section, we determine the exact value for
   {\rm min-seed}$(G,\theta_>)$ where $G$ is a honeycomb mesh
   network with strict majority threshold function $\theta_>$.
   %
   %  Honeycomb Meshes
   %
   The {\em honeycomb mesh} of size $t$ (see \cite{Stojmenovic1997}
   for a comprehensive introduction to this class of graphs
   and their variants),
   denoted by HM$_t$ is defined inductively as follows:
   HM$_1$ is a hexagon. Honeycomb mesh HM$_t$ of size $t>1$
   is obtained from HM$_{t-1}$ by adding a layer of hexagons
   around the boundary of HM$_{t-1}$.
   The number of vertices and edges of HM$_t$ are $6t^2$ and
   $9t^2-3t$, respectively.
   The edges of HM$_t$ are in $3$ different directions.
   See Figure 1 for
   examples of HM$_t$ when $t=1,2,3$, where
   the point ${O}$ of HM$_t$ is called the {\em centre} of
   the honeycomb mesh.
   Through ${O}$ one can draw three lines perpendicular to
   the three edge directions and name them as $\alpha, \beta, \gamma$ axes.
   These three axes will be used in Section \ref{GHRT} to define the
   honeycomb torus network introduced in \cite{Stojmenovic1997}.

   Every honeycomb mesh has a nice drawing as shown in Figure 2.
   We call this kind of drawing the {\em castle drawing}.
   In Figure 2,
   we show an addressing scheme to describe
   the vertices of a honeycomb mesh
   %we also describe a suitable addressing scheme
   %for the vertices of HM$_t$
   which will be used in the proof of
  Theorem \ref{main-min-seed-HM_t}.

 %%%%%%%%%%%%%%%%%%%%%%%%%%%%%

 %  Figure 1: HM123.eps

 %%%%%%%%%%%%%%%%%%%%%%%%%%%%%
 \includegraphics[scale=0.55]{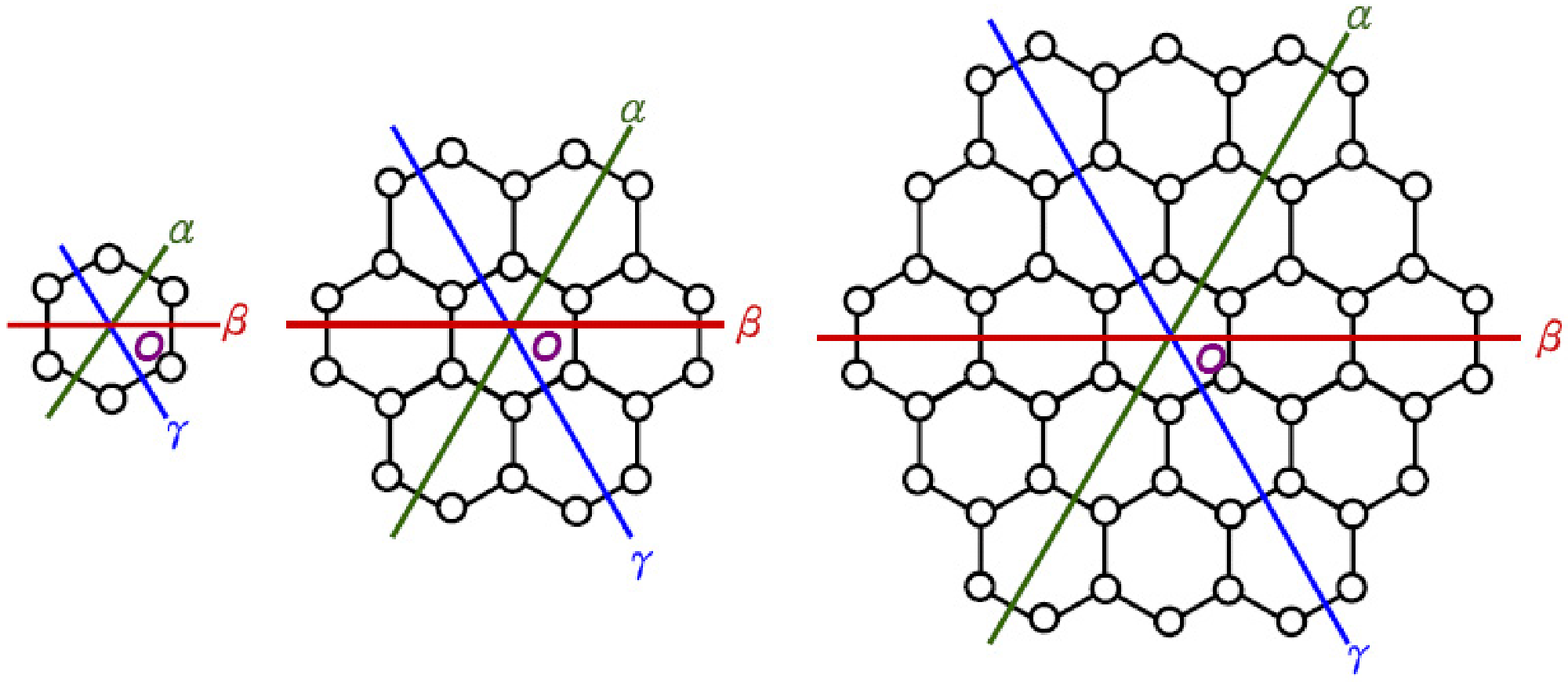}%\\~\\
 \begin{description}
    \item[Figure 1.] HM$_1$ (left), HM$_2$ (middle) and
    HM$_3$ (right).
 \end{description}

 %%%%%%%%%%%%%%%%%%%%%%%%%%%%%

 %  Figure 2: castle-HM1-2-3.eps  (castle drawing)

 %%%%%%%%%%%%%%%%%%%%%%%%%%%%%
 \includegraphics[scale=0.35]{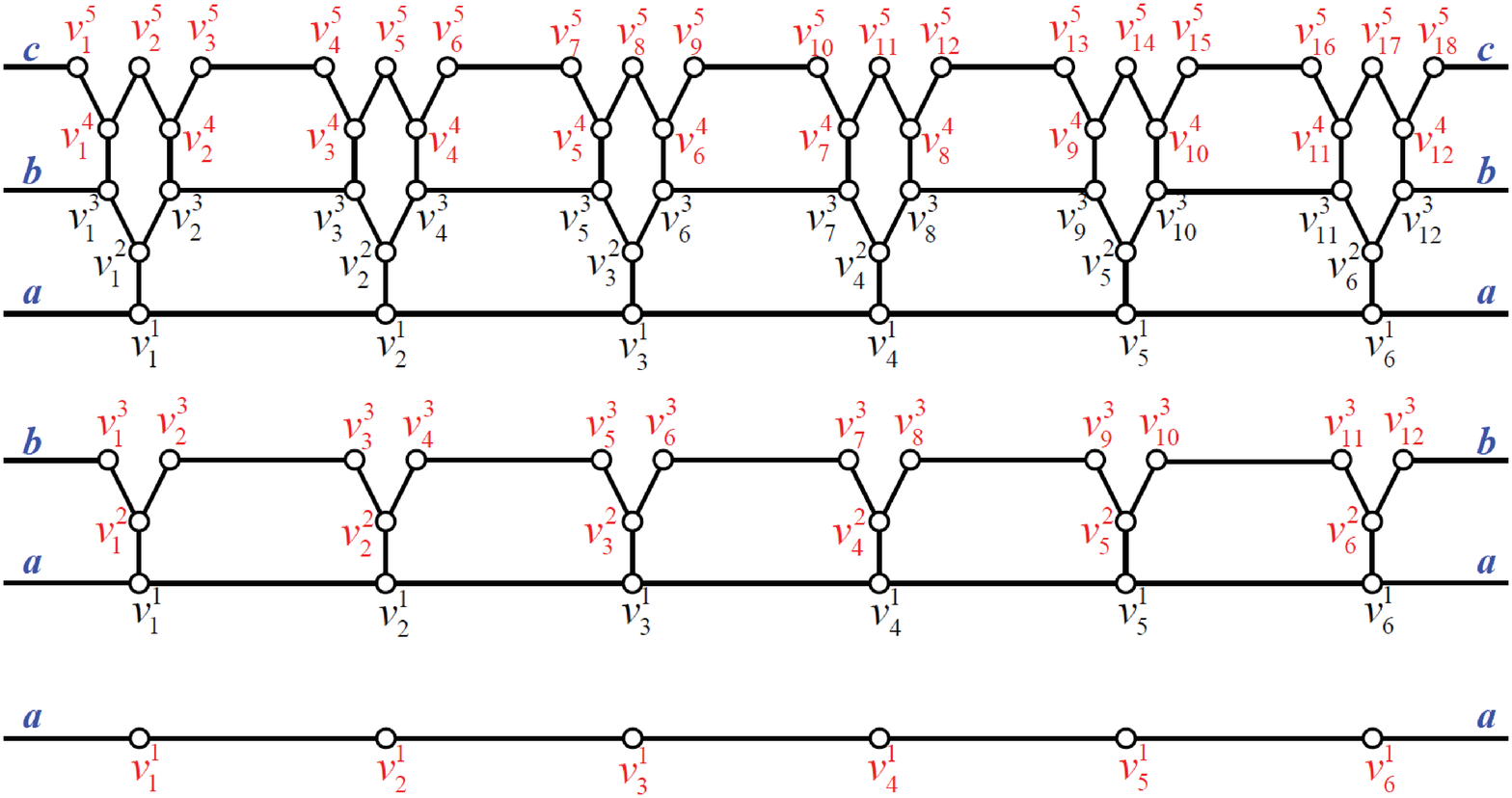}%\\~\\
 \begin{description}
    \item[Figure 2.] The castle drawing
    of a honeycomb mesh
    with a coordinate system on it: HM$_1$ (lower), HM$_2$ (middle) and
    HM$_3$ (upper). In each graph, the same edge is given the same label,
    that is, $v_1^1v_6^1, v_1^3v_{12}^3, v_1^5v_{18}^5$
    are edges.
 \end{description}

  %
  %
  %  min-seed for HM_t
  %
  \begin{theorem}
  \label{main-min-seed-HM_t}
  {\rm min-seed}$({\rm HM}_t,\theta_>)={(3t^2+3t)/2}$.
  \end{theorem}
  \begin{proof}
  Let $G={\rm HM}_t$ and $\theta_{\min}=\min\{\theta_>(v):v\in V(G)\}$.
  Obviously, {\rm min-seed}$(G,\theta_>)={\mbox{ min-seed}}(G,2)$.
  Since HM$_t$ has $6t^2$ vertices and $9t^2-3t$ edges, by
  the result min-seed$(G,\theta_>)\geq {|V(G)|\theta_{\min}-|E(G)|\over \theta_{\min}}$
  presented in Lemma \ref{lowerbound-for-target-set}, we see at once that
  {\rm min-seed}$(G,2)\geq {(3t^2+3t)/2}$.

   Next we will show that {\rm min-seed}$(G,2)\leq {(3t^2+3t)/2}$
   by giving a target set $S$ of size ${(3t^2+3t)/2}$
   which influences all vertices of $V(G)\setminus S$ in $(G,2)$.
   Consider the castle drawing of $G$
   with an addressing scheme on vertices, as shown in Figure 2.
   For each positive integer $i$, define that
   $$V_i=\left\{ \begin{array}{ll}
                   \{v_1^i,v_2^i,\ldots,v_{3i}^i\} &\mbox{ if $i$ is even,}\\
                   \{v_1^i,v_2^i,\ldots,v_{3i+3}^i\} &\mbox{ if $i$ is odd.} \\
                 \end{array} \right.$$
   It can be seen that $V(G)=\cup_{i=1}^{2t-1}V_i$.
   Consider $S=\cup_{k=1}^t \{v_1^{2k-1},v_3^{2k-1},v_5^{2k-1},\ldots,v_{6k-1}^{2k-1}\}$
   as a target set for $(G,2)$ (see Figure 3 for a graphical illustration of this
   target set $S$).
   It is easy to check that $S$ can influence all vertices of
   $V(G)\setminus S$ by using the convinced sequence
   $\beta=\alpha_1\sqcup \beta_2$, where
   $\alpha_1=\{v_2^1,v_4^1,v_6^1\}$ and
   $\beta_2=\sqcup_{i=2}^t ([v_1^{2i-2},v_2^{2i-2},v_3^{2i-2},
        \ldots,v_{6(i-1)}^{2i-2}]
             \sqcup
             [v_2^{2i-1},v_4^{2i-1},v_6^{2i-1},
        \ldots,v_{6i}^{2i-1}])$
   (see Figure 1 in Appendix for a graphical illustration of this
   convinced sequence $\beta$).
  Since the cardinality of $S$ is
  $\sum_{k=1}^{t} {3k}= {3t(t+1)\over 2}$,
  we have {\rm min-seed}$({\rm HM}_t,2)\leq {(3t^2+3t)/2}$,
  which completes the proof of the theorem.
  \end{proof}

 \bigskip
 %%%%%%%%%%%%%%%%%%%%%%%%%%%%%

 %  Figure 3: seed-HM1-2-3.eps

 %%%%%%%%%%%%%%%%%%%%%%%%%%%%%
 \includegraphics[scale=0.58]{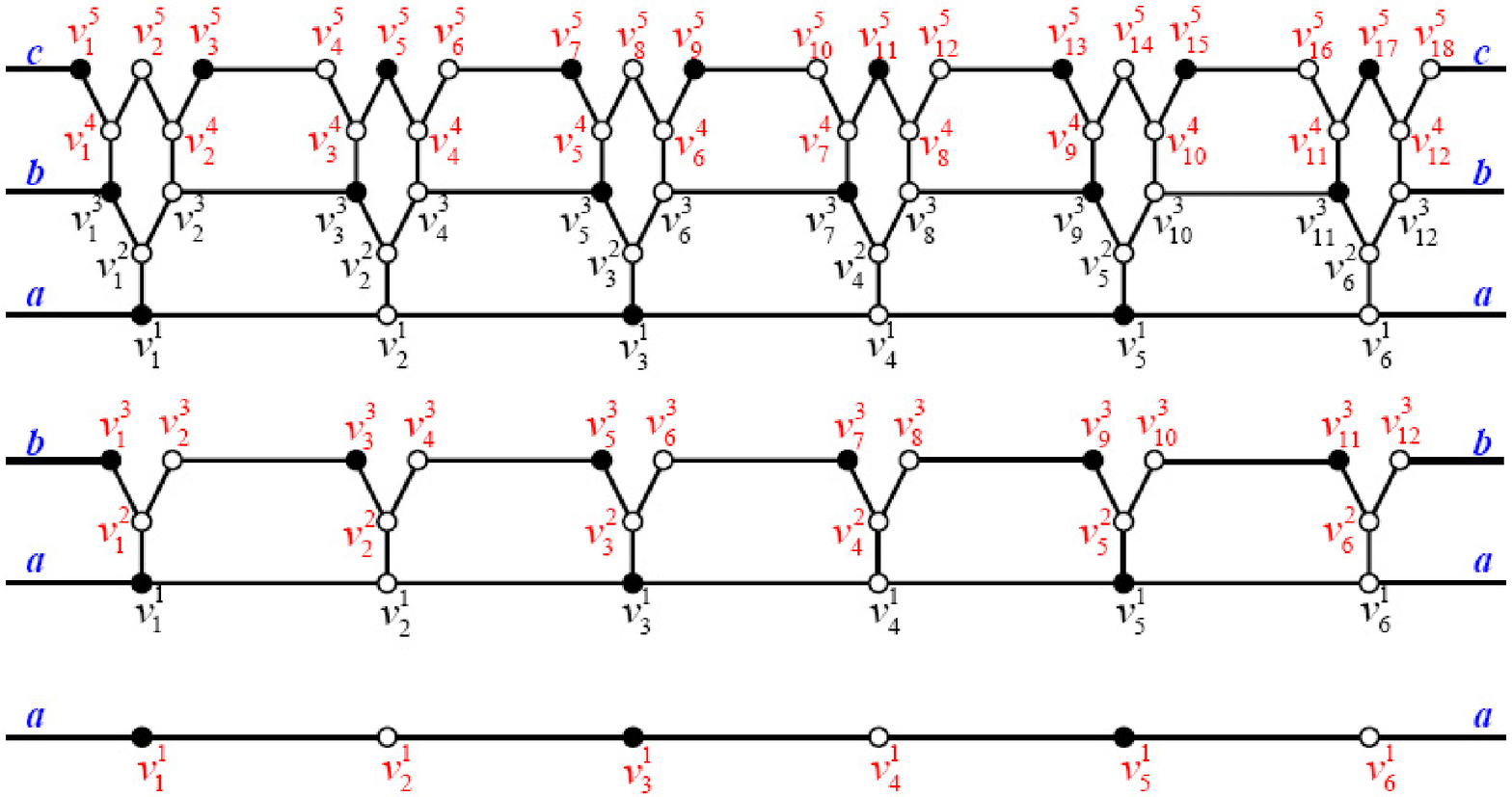}%\\~\\
 \begin{description}
    \item[Figure 3.]
    HM$_1$ (lower), HM$_2$ (middle) and
    HM$_3$ (upper),
    where the target set $S$ is the set of all black vertices.
 \end{description}

   %%%%%%%%%%%%%%%%%%%%%%
   %
   %
   %   Section 4:
   %
   %
   %%%%%%%%%%%%%%%%%%%%%%
\section{Generalized honeycomb rectangular torus}
\label{GHRT}

  In this section, under strict majority threshold model,
  we study the problem of computing
  optimal target sets for three
  well-known honeycomb tori:
  honeycomb torus,
  honeycomb rectangular torus, and
  honeycomb rhombic torus.
  Actually, we will tackle this problem by considering
  a slightly more general class
  of network topologies called generalized honeycomb rectangular torus.

   %
   %  Honeycomb Torus
   %
  The {\em honeycomb torus} of size $t$ introduced in \cite{Stojmenovic1997},
  denoted by HT$_t$,
  is obtained from a honeycomb mesh of size $t$ by
  joining the pairs of degree $2$ vertices in HM$_t$
  that are mirror symmetric with respect to the three axes
  $\alpha,\beta,\gamma$ of the HM$_t$ (see Figure 1 for the
  three axes of a honeycomb mesh).
  Figure 4 shows how to wraparound HM$_1$, HM$_2$ and HM$_3$
  to obtain HT$_1$, HT$_2$ and HT$_3$, respectively.

 %%%%%%%%%%%%%%%%%%%%%%%%%%%%%
 %  Figure 4: HT1-2-3.eps
 %%%%%%%%%%%%%%%%%%%%%%%%%%%%%
 \includegraphics[scale=0.7]{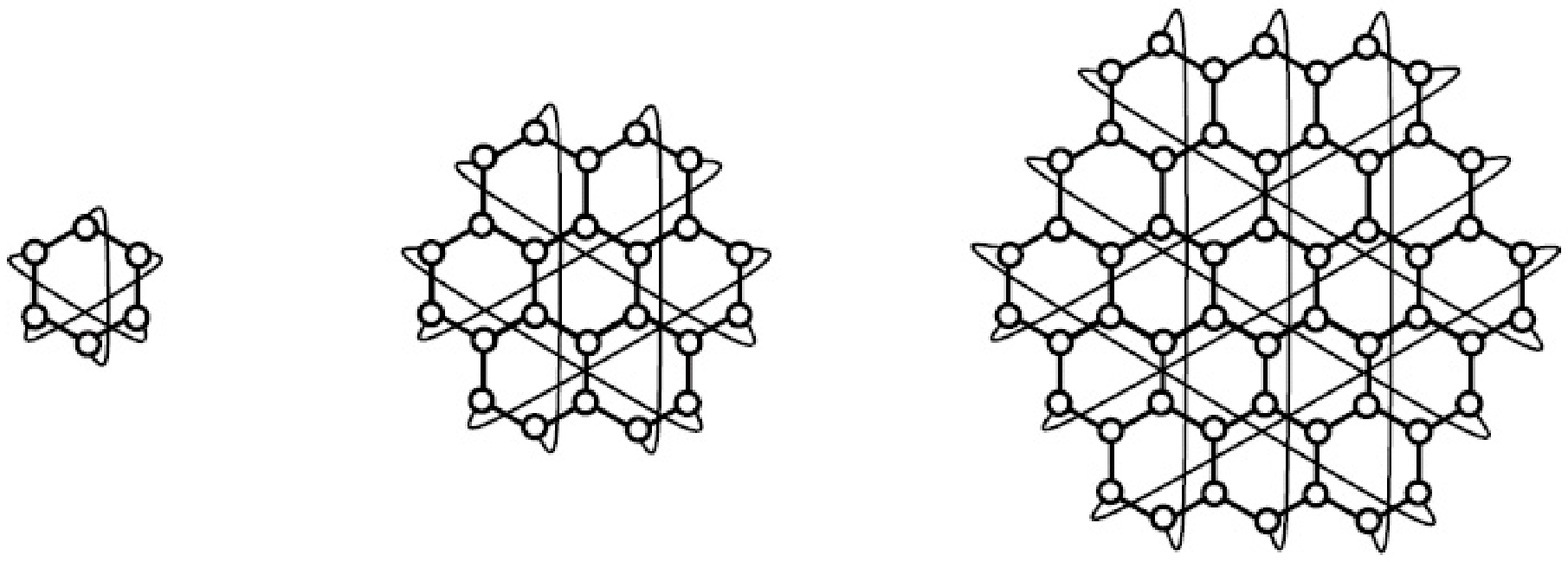}%\\~\\
 \begin{description}
    \item[Figure 4.] HT$_1$ (left), HT$_2$ (middle) and HT$_3$ (right).
 \end{description}

 %%%%%%%%%%%%%%%%%%%%%%%%%%%%%%%%%%%%%%%%%%
 %     honeycomb rectangular torus
 %%%%%%%%%%%%%%%%%%%%%%%%%%%%%%%%%%%%%%%%%%
 Let $m$ and $n$ be positive
 even integers such that $n\geq 4$.
 The {\em honeycomb rectangular torus} ${\rm HReT}(m,n)$,
 introduced by Stojmenovic \cite{Stojmenovic1997}
 (see also \cite{Cho2003-GHT,Parhami2001}), is the graph
 with the
 vertex set $\{(i,j): 0\leq i<m, 0\leq j<n\}$
 such that  $(i,j)$ and $(k,\ell)$ are adjacent if and only if they
 satisfy one of the following conditions:

 \begin{enumerate}
    \item $i=k$ and $j=\ell\pm 1 \pmod n$;
  %  \item $j=\ell$ and $k=i+1\pmod m$ if $i+j$ is odd; and
    \item $j=\ell$ and $k=i-1 \pmod m$ if $i+j$ is even.
 \end{enumerate}
 For example, consider Figure 5(left) which depicts ${\rm
 HReT}(4,6)$.
 Note that our notation for ${\rm HReT}(m,n)$
  is slightly different from the one used by
 Stojmenovic in \cite{Stojmenovic1997}.

 %%%%%%%%%%%%%%%%%%%%%%%%%%%%%%%%%%%%%%%%%%%%%%%
 %  Figure 5: HReT(4-6)HRoT(5-6)GHT(4-6-2).eps
 %%%%%%%%%%%%%%%%%%%%%%%%%%%%%%%%%%%%%%%%%%%%%%%
 \begin{center}
 \includegraphics[scale=0.7]{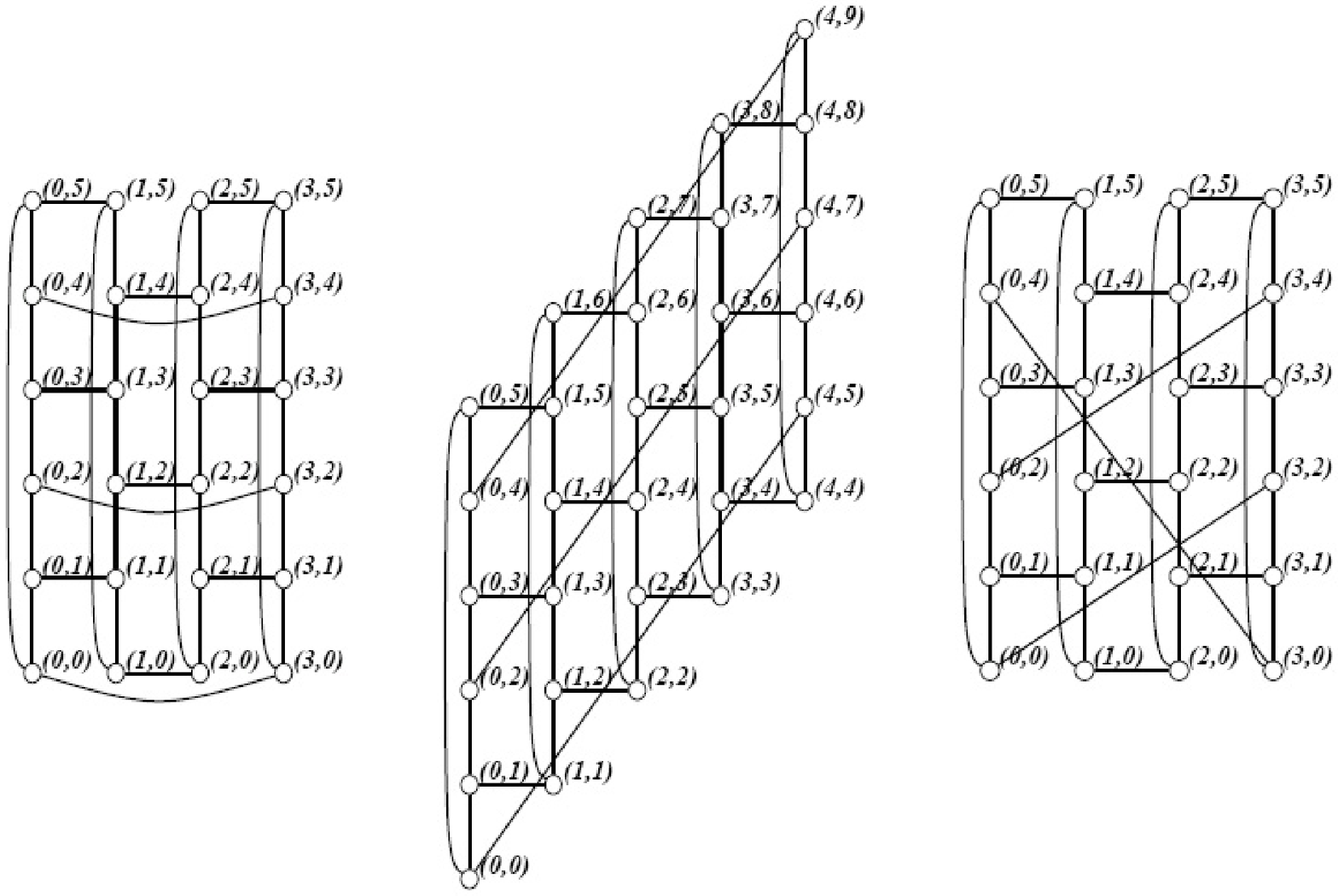}%\\~\\
 \end{center}
 \begin{description}
    \item[Figure 5.]
    ${\rm HReT}(4,6)$ (left), ${\rm HRoT}(5,6)$ (middle), and
    ${\rm GHT}(4,6,2)$ (right).
 \end{description}

 %%%%%%%%%%%%%%%%%%%%%%%%%%%%%%%%%%%%%%%%%%
 %     honeycomb rhombic torus
 %%%%%%%%%%%%%%%%%%%%%%%%%%%%%%%%%%%%%%%%%%
 Let $m$ and $n$ be positive integers such that $n$ is even.
 The {\em honeycomb rhombic torus} ${\rm HRoT}(m,n)$,
 introduced by Stojmenovic \cite{Stojmenovic1997}
 (see also \cite{Cho2003-GHT,Yang2004}), is the graph
 with the vertex set $\{(i,j): 0\leq i<m, 0\leq j-i<n\}$
 such that  $(i,j)$ and $(k,\ell)$ are adjacent if and only if they
 satisfy one of the following conditions:
 \begin{enumerate}
    \item $i=k$ and $j=\ell\pm 1 \pmod n$;
    \item $j=\ell$ and $k=i-1$ if $i+j$ is even; and
    \item $i=0$, $k=m-1$, and $\ell=j+m$ if $j$ is even.
 \end{enumerate}
  For example, consider Figure 5(middle) which depicts
  ${\rm HRoT}(5,6)$.
   Note that our notation ${\rm HRoT}(m,n)$ for a
   honeycomb rhombic torus
   is different from the one used in \cite{Stojmenovic1997,Yang2004}.

 %%%%%%%%%%%%%%%%%%%%%%%%%%%%%%%%%%%%%%%%%%%%%
 %     Generalized honeycomb rectangular torus
 %%%%%%%%%%%%%%%%%%%%%%%%%%%%%%%%%%%%%%%%%%%%%
 In \cite{Cho2003-GHT} Cho and Hsu introduced a
 class of generalized honeycomb tori which
 cover the three honeycomb tori mentioned above.
 Let $m$ and $n$ be positive integers such that $n\geq 4$ is even.
 Let $d$ be any nonnegative integer such that $m-d$ is an even number.
 The {\em generalized honeycomb rectangular torus}
 (or {\em generalized honeycomb torus}), denoted by ${\rm
 GHT}(m,n,d)$ and
 proposed by Cho and Hsu \cite{Cho2003-GHT},
 is the graph
 with the
 vertex set $\{(i,j): 0\leq i<m, 0\leq j<n\}$
 such that  $(i,j)$ and $(k,\ell)$ are adjacent if and only if they
 satisfy one of the following conditions:

 \begin{enumerate}
    \item $i=k$ and $j=\ell\pm 1 \pmod n$;
    \item $j=\ell$ and $k=i-1$ if $i+j$ is even; and
    \item $i=0$, $k=m-1$, and $\ell=j+d \pmod n$ if $j$ is even.
 \end{enumerate}
  For example, Figure 5(right) depicts a ${\rm GHT}(4,6,2)$.
  We remark that, in \cite{Alspach-Dean2009}, the authors call
  ${\rm GHT}(m,n,d)$ the {\em honeycomb toroidal graph}.

  Now, given a generalized honeycomb rectangular
  torus $G$, in the proof of Theorem \ref{main-min-seed-GHT(m,n,d)},
  we shall show how to compute an optimal target set for $G$ under
  strict majority threshold model.

  %
  %
  %  min-seed for GHT(m,n,d)
  %
  \begin{theorem}
  \label{main-min-seed-GHT(m,n,d)}
  If $G$ is a generalized honeycomb rectangular torus ${\rm GHT}(m,n,d)$,
  then $\mbox{\rm min-seed}(G,\theta_>)$$=\left\lceil{(mn+2 )/4}\right\rceil$.
  \end{theorem}
  \begin{proof}
  Let $G={\rm GHT}(m,n,d)$ and
  $\delta=\max\{d_G(v)-\theta_>(v): v\in V(G)\}$.
  Let $\Delta$ be the maximum degree of $G$.
  Obviously, $G$ is a $3$-regular graph. It follows that
  {\rm min-seed}$(G,\theta_>)={\mbox{ min-seed}}(G,2)$.
  Since $G$ has $mn$ vertices and ${3mn\over 2}$ edges, by
  the result min-seed$(G,\theta_>)\geq
  {|E(G)|-(|V(G)|-1)\delta \over \Delta-\delta}$
  presented in Lemma \ref{lowerbound-for-target-set}, we see at once that
  {\rm min-seed}$(G,2)\geq \left\lceil{(mn+2)/ 4}\right\rceil$.

  Next, we shall prove that
  {\rm min-seed}$(G,2)\leq \left\lceil{(mn+2)/ 4}\right\rceil$
  by giving a target set $S$ for $(G,2)$ which influences all
  vertices in $V(G)\setminus S$ and has cardinality
  $\left\lceil{(mn+2)/ 4}\right\rceil$.
  Note that $n\geq 4$ is even.
  We let $n=4t+r$, where $t$ is a positive integer and $r\in \{0,2\}$.
  The proof is divided into three cases, according to the parity of $m$
  and the value of $r$.

  {\bf Case 1.} $m$ is even.
  Let $S_1=\cup_{j=0}^{(n-4)/2}\{(0,2j),(2,2j),(4,2j),\ldots,(m-2,2j)\}$
  and $S_2=\{(1,n-1),(3,n-1),(5,n-1),\ldots, (m-1,n-1)\}$.
  Consider
   $S=S_1\cup S_2\cup \{(0,n-2)\}$
   as a target set for $(G,2)$
   (see Figure 6 for a graphical illustration of $S$).
   Note that, in this case, $d$ is even.
   By the definition of ${\rm GHT}(m,n,d)$
   and by the choice of $S$, it can be seen that
   if $\ell$ is even, then the
   vertex $(m-1,\ell)$ is adjacent to a vertex $(0,j)$ in $S$
   such that $j$ is even.
   With this observation,
   it is straightforward to check that $S$ can influence all vertices of
   $V(G)\setminus S$ by using the convinced sequence
   $\alpha=\alpha_1 \sqcup \alpha_2 $
   (see Figure 2 in Appendix for a graphical illustration of this
   convinced sequence $\alpha$), where
   $\alpha_1=[(0,n-1),(0,n-3),(0,n-5),\ldots,(0,1)]
     \sqcup  [(1,0),(1,1),(1,2),\ldots,(1,n-2)]$ and
   $\alpha_2=\sqcup_{i=1}^{{m \over 2}-1} ([(2i,n-1),(2i,n-2)]
     \sqcup [(2i,n-3),(2i,n-5),(2i,n-7),\ldots,(2i,1)]
     \sqcup [(2i+1,0),(2i+1,1),(2i+1,2),\ldots,(2i+1,n-2)])$.
   Since
   $|S|={mn\over 4}+1=\left\lceil{(mn+2)/ 4}\right\rceil$,
   we obtain the desired inequality
   {\rm min-seed}$(G,2)\leq \left\lceil{(mn+2)/ 4}\right\rceil$.

   {\bf Case 2.} $m$ is odd and $r=0$.
  Let
  $S_1=\cup_{j=0}^{(n-4)/2}\{(0,2j),(2,2j),(4,2j),\ldots,(m-3,2j)\}$,
  $S_2=\{(1,n-1),(3,n-1),(5,n-1),\ldots, (m-2,n-1)\}$,
  and $S_3=\{(m-1,0),(m-1,4),(m-1,8),\ldots,(m-1,n-4)\}$.
  Consider
   $S=S_1\cup S_2\cup S_3\cup \{(0,n-2)\}$
   as a target set for $(G,2)$
   (see Figure 7 for a graphical illustration of $S$).
   Note that, in this case, $d$ is odd.
   By the definition of ${\rm GHT}(m,n,d)$
   and by the choice of $S$, we see that
   if $\ell$ is odd, then the
   vertex $(m-1,\ell)$ is adjacent to a vertex $(0,j)$ in $S$
   such that $j$ is even.
   With the above in mind, it is straightforward to check that $S$
   can influence all vertices of
   $V(G)\setminus S$ by using the convinced sequence
   $\alpha=\alpha_1 \sqcup \alpha_2 \sqcup \alpha_3 \sqcup \alpha_4$
   (see Figure 3 in Appendix for a graphical illustration of this
   convinced sequence $\alpha$), where
   $\alpha_1=[(m-1,1),(m-1,3),(m-1,5),\ldots,(m-1,n-1)]$,
   $\alpha_2=[(m-1,2),(m-1,6),(m-1,10),(m-1,14),\ldots,(m-1,n-6),(m-1,n-2)]$,
   $\alpha_3=[(0,n-1),(0,n-3),(0,n-5),\ldots,(0,1)]
     \sqcup  [(1,0),(1,1),(1,2),\ldots,(1,n-2)]$, and
   $\alpha_4=\sqcup_{i=1}^{{m-3 \over 2}} ([(2i,n-1),(2i,n-2)]
     \sqcup [(2i,n-3),(2i,n-5),(2i,n-7),\ldots,(2i,1)]
     \sqcup [(2i+1,0),(2i+1,1),(2i+1,2),\ldots,(2i+1,n-2)])$.
   Since $n\equiv 0 \pmod 4$, we have
   $|S|={mn\over 4}+1=\left\lceil{(mn+2)/ 4}\right\rceil$, and hence
   {\rm min-seed}$(G,2)\leq \left\lceil{(mn+2)/ 4}\right\rceil$.

   {\bf Case 3.} $m$ is odd and $r=2$.
   In the following proof, the second coordinate of a vertex $(a,b)$ in $G$
   is read modulo $n$, for example
   we have $(m-1,d+4t-1)=(m-1,d-3)$.
  Let
  $S_1=\cup_{j=0}^{(n-4)/2}\{(0,2j),(2,2j),(4,2j),\ldots,(m-3,2j)\}$,
  $S_2=\{(1,n-1),(3,n-1),(5,n-1),\ldots, (m-2,n-1)\}$,
  and $S_3=\{(m-1,d-1),(m-1,d+3),(m-1,d+7),\ldots,(m-1,d+4t-1)\}$.

  By the definition of ${\rm GHT}(m,n,d)$, we see that
  $(m-1,d),(m-1,d+2),(m-1,d+4),\ldots,(m-1,d+4t-2)$ are adjacent to
  vertices $(0,0),(0,2),(0,4),\ldots, (0, 4t-2)$, respectively, and
  the vertex $(m-1,d-2)$ is adjacent to both $(m-1,d-1)$ and $(m-1,d-3)$.
  Note that $\{(0,0),(0,2),(0,4),\ldots, (0, 4t-2)\}\subseteq S_1$
  and $\{(m-1,d-1),(m-1,d-3)\}\subseteq S_3$.
  Consider
   $S=S_1\cup S_2\cup S_3$
   as a target set for $(G,2)$
   (see Figure 8 for a graphical illustration of $S$).
   By the above observation and by the choice of $S$,
   it is straightforward to check that $S$ can influence all vertices of
   $V(G)\setminus S$ by using the convinced sequence
   $\alpha=\alpha_1 \sqcup \alpha_2 \sqcup \alpha_3$
   (see Figure 4 in Appendix for a graphical illustration of this
   convinced sequence $\alpha$), where
   $\alpha_1=[(m-1,d-2),(m-1,d),(m-1,d+2),(m-1,d+4),\ldots,(m-1,d+4t-2)]$,
   $\alpha_2=[(m-1,d+1),(m-1,d+5),(m-1,d+9),(m-1,d+13),\ldots,(m-1,d+4t-3)]$,
   and
   $\alpha_3=\sqcup_{i=0}^{{(m-3)/ 2}} ([(2i,n-1),(2i,n-2)]
     \sqcup [(2i,n-3),(2i,n-5),(2i,n-7),\ldots,(2i,1)]
     \sqcup [(2i+1,0),(2i+1,1),(2i+1,2),\ldots,(2i+1,n-2)])$.
   Since
   $|S|={n(m-1)\over 4}+t+1={mn+2 \over 4}=\left\lceil{(mn+2)/ 4}\right\rceil$, we see that
   {\rm min-seed}$(G,2)\leq \left\lceil{(mn+2)/ 4}\right\rceil$.
   This completes the proof of the theorem.
  \end{proof}

 %%%%%%%%%%%%%%%%%%%%%%%%%%%%%
 %  Figure 6: seed-GHT(6-8-4).eps
 %%%%%%%%%%%%%%%%%%%%%%%%%%%%%
 \begin{center}
 \includegraphics[scale=0.55]{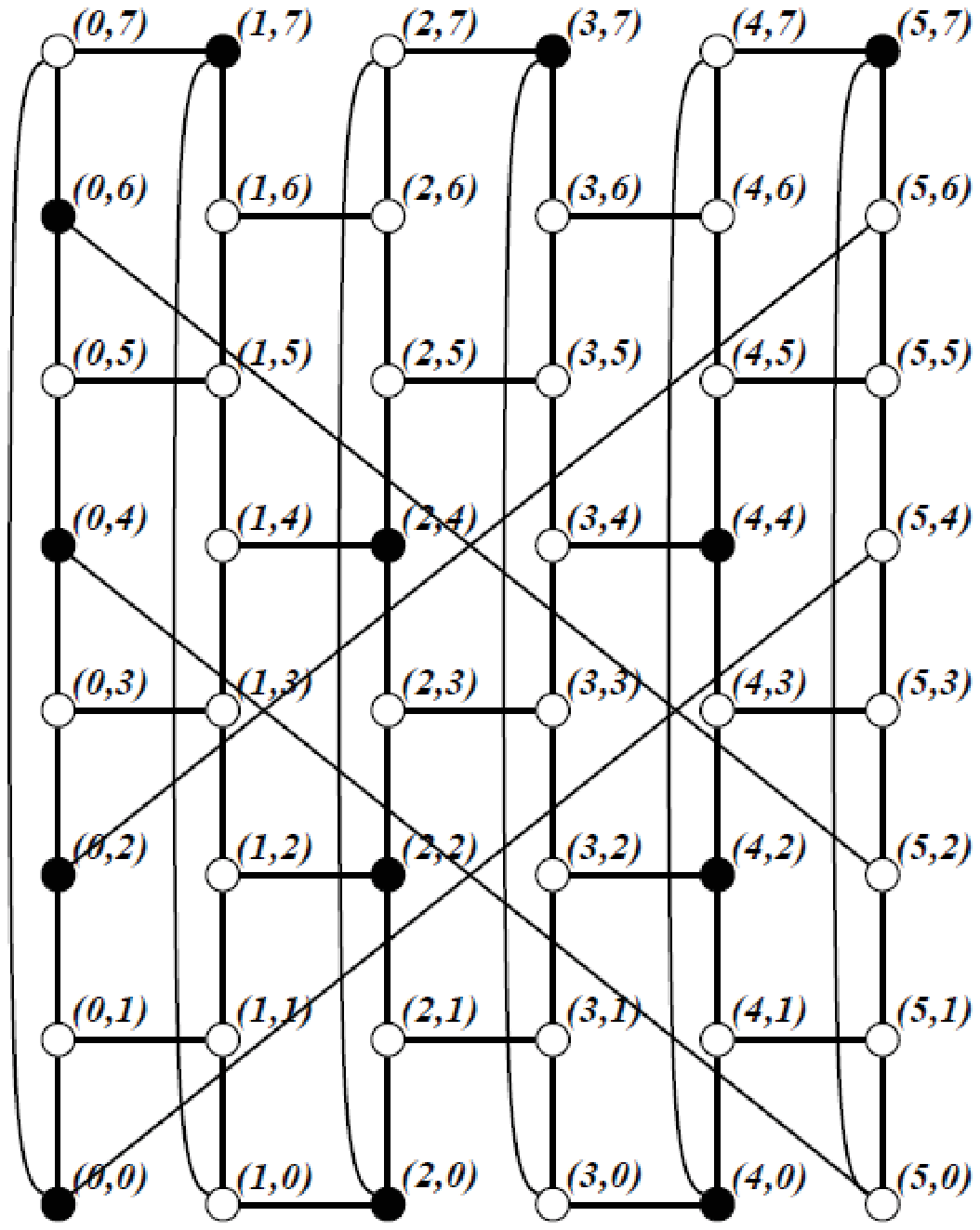}%\\~\\
 \end{center}
 \begin{description}
    \item[Figure 6.] ${\rm GHT}(6,8,4)$
    where the target set $S$ is the set of all black vertices.
 \end{description}
  %%%%%%%%%%%%%%%%%%%%%%%%%%%%%
 %  Figure 7: seed-GHT(9-8-5).eps
 %%%%%%%%%%%%%%%%%%%%%%%%%%%%%
 \begin{center}
 \includegraphics[scale=0.55]{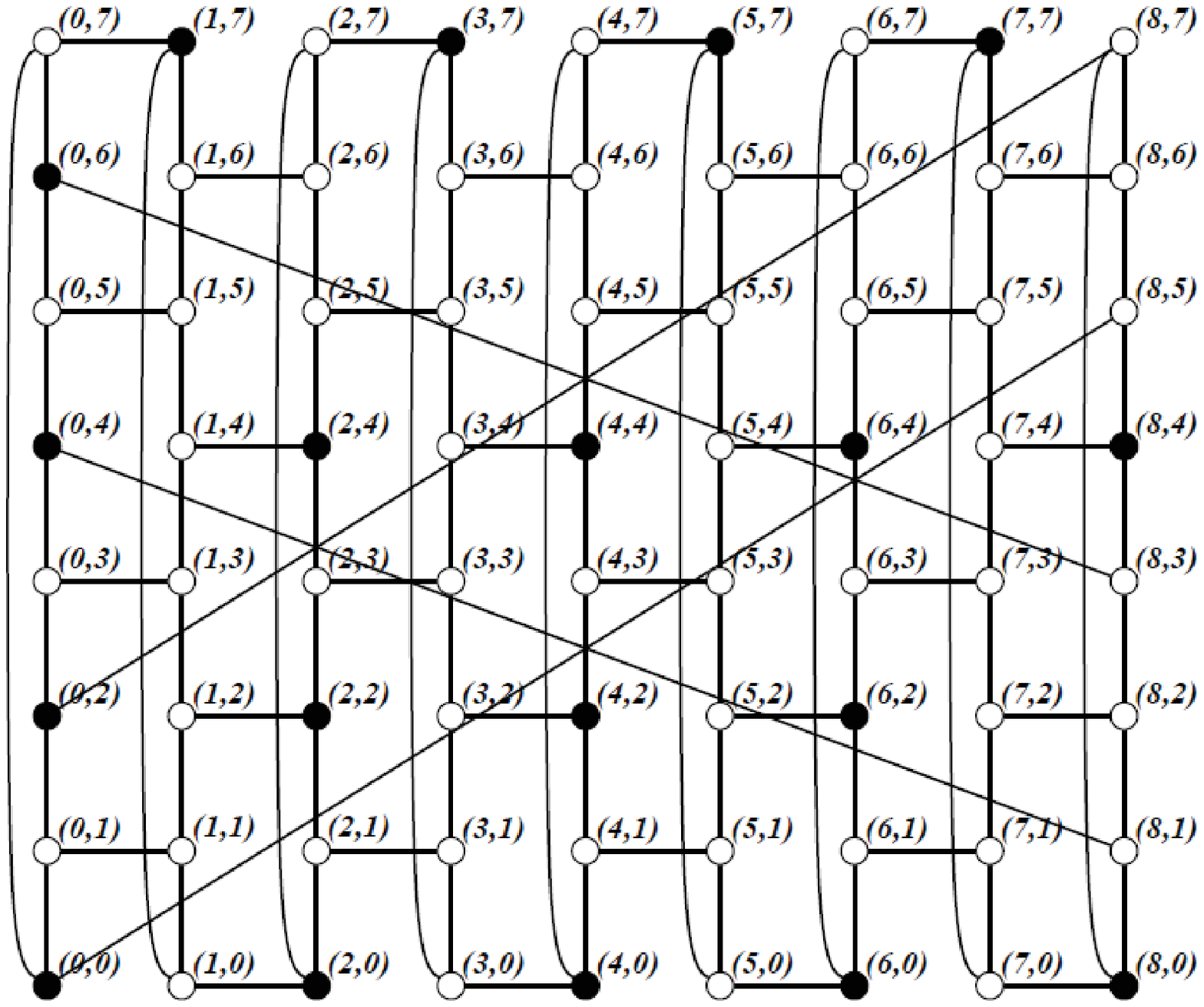}%\\~\\
 \end{center}
 \begin{description}
    \item[Figure 7.] ${\rm GHT}(9,8,5)$
    where the target set $S$ is the set of all black vertices.
 \end{description}
  %%%%%%%%%%%%%%%%%%%%%%%%%%%%%
 %  Figure 8: seed-GHT(9-10-5).eps
 %%%%%%%%%%%%%%%%%%%%%%%%%%%%%
 \begin{center}
 \includegraphics[scale=0.55]{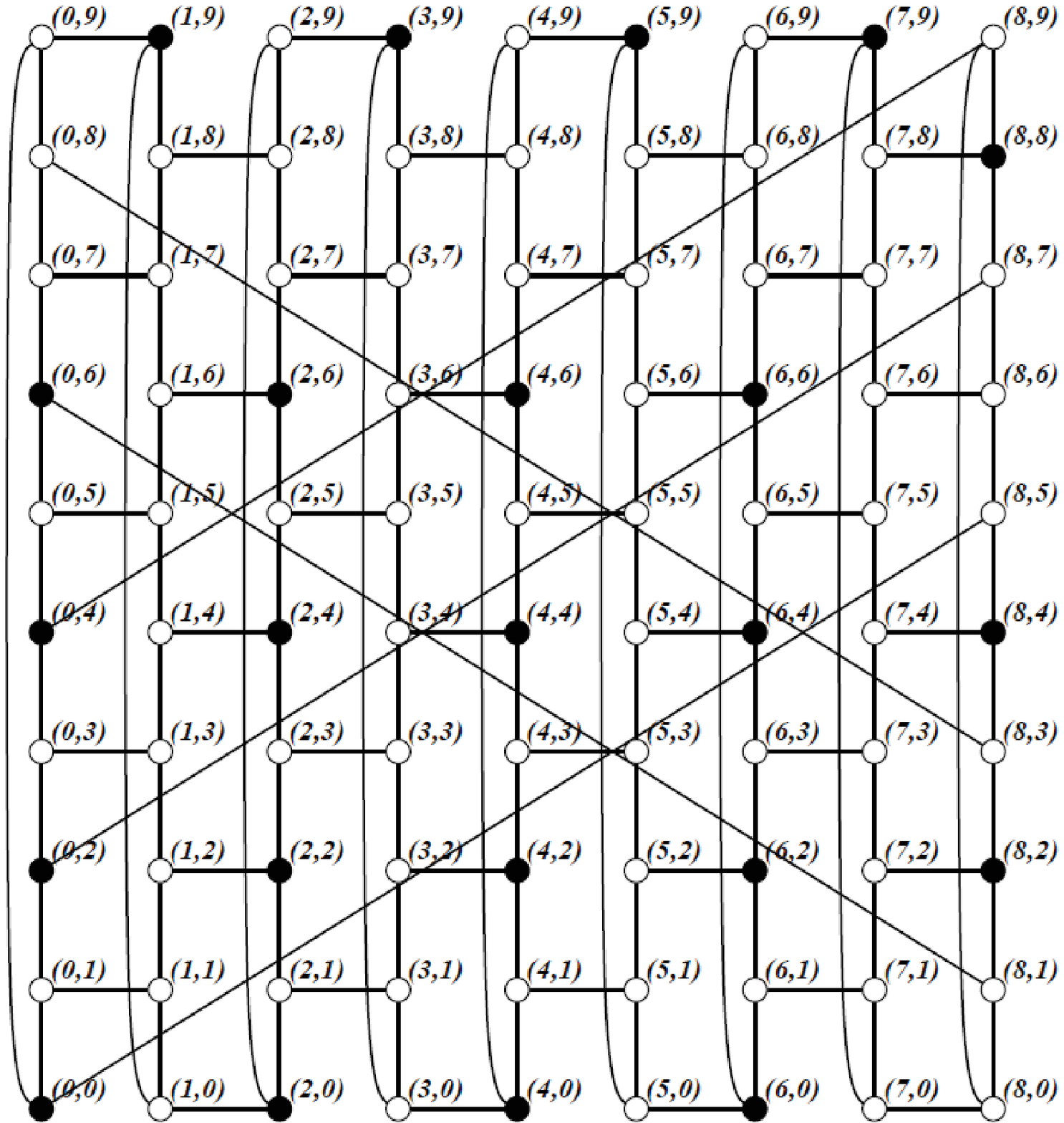}%\\~\\
 \end{center}
 \begin{description}
    \item[Figure 8.] ${\rm GHT}(9,10,5)$
    where the target set $S$ is the set of all black vertices.
 \end{description}

  From the definitions of the honeycomb rectangular torus,
  the honeycomb rhombic torus, and the generalized honeycomb
  rectangular torus, it can readily be seen that
  ${\rm HReT}(m,n)$ is isomorphic to ${\rm GHT}(m,n,0)$
  and
  ${\rm HRoT}(m,n)$ is isomorphic to ${\rm GHT}(m,n,m\, (\bmod\,
  n))$.
  In \cite{Cho2003-GHT}, Cho and Hsu proved that
  the honeycomb torus of size $t$
  is isomorphic to ${\rm GHT}(t,6t,3t)$.
  Now the following corollary follows immediately from
  Proposition 1 of \cite{Dreyer+Roberts},
  Theorem \ref{main-min-seed-GHT(m,n,d)} and the above
  discussion.

  %
  %
  %  min-seed for HT(n), HReT(m,n), HRoT(m,n)
  %
  \begin{corollary}
  \label{main-corollary}
  $(1)$ If $G$ is a generalized honeycomb rectangular torus
  ${\rm GHT}(m,n,d)$,
  then the decycling number $\nabla(G)=\left\lceil{(mn+2 )/4}\right\rceil$.
  $(2)$ If $G$ is a honeycomb torus ${\rm HT}_t$ then
  {\rm min-seed}$(G,\theta_>)=\nabla(G)= \lceil{(3t^2+1)/ 2}\rceil$.
  $(3)$ If $G$ is a honeycomb rectangular torus ${\rm HReT}(m,n)$ or a
  honeycomb rhombic torus ${\rm HRoT}(m,n)$, then
  {\rm min-seed}$(G,\theta_>)=\nabla(G)= \lceil{(mn+2)/ 4}\rceil$.
  \end{corollary}

   %%%%%%%%%%%%%%%%%%%%%%
   %
   %
   %   Section 5:
   %
   %
   %%%%%%%%%%%%%%%%%%%%%%
   \section{Hexagonal grids}
   \label{sec-hexagonal-grids}
   In this section, under strict majority threshold model,
  we study the problem of computing
  optimal target sets for a graph $G$ which has an underlying
  hexagonal (or honeycomb) grid structure.
  Let $m$ and $n$ be two integers such that $m\geq 2$, $n\geq 4$, and $n$ even.
   %
   %     planar hexagonal grids
   %
  An {\em $m$ by $n$ planar hexagonal grid}, denoted by
  ${\rm PHG}(m,n)$, consists of an array of
  $n$ rows of $m$ vertices $(x,y)$, with $0\leq x\leq m-1$, $0\leq y\leq n-1$,
  arranged on a standard Cartesian plane such that each vertex $(x,y)$ is adjacent to
  $(x,y+1)$ and, if $y$ is even (zero is considered to be even),
  also adjacent to $(x+1,y+1)$, provided that each coordinate
  is within its allowed range and no vertex of degree one is generated.
  As an example, PHG(5,8) is depicted in Figure 9.

   %
   %     cylindrical hexagonal grids
   %
  An {\em $m$ by $n$ cylindrical hexagonal grid} ${\rm CHG}(m,n)$
  is obtained from the $m$ by $n$ planar hexagonal grid ${\rm PHG}(m,n)$
  by adding the edges from $(m-1, y)$ to $(0,y+1)$ for any even $y$.
  In other words, a ${\rm CHG}(m,n)$ is defined the same as
  a ${\rm PHG}(m,n)$
  except that for each vertex $(x,y)$ the addition in the first
  coordinate is taken modulo $m$.
  As an example, CHG(5,8) is depicted in Figure 9.

   %
   %     toroidal hexagonal grids
   %
  An {\em $m$ by $n$ toroidal hexagonal grid}, denoted by
  ${\rm THG}(m,n)$, is defined the same as
  a ${\rm PHG}(m,n)$
  except that for each vertex $(x,y)$ addition in the first
  coordinate is taken modulo $m$ and
  addition in the second coordinate is taken modulo $n$.
  As an example, THG(5,8) is depicted in Figure 9.

 %%%%%%%%%%%%%%%%%%%%%%%%%%%%%
 %  Figure 9: P-C-T-HG5x8.eps
 %%%%%%%%%%%%%%%%%%%%%%%%%%%%%
 \includegraphics[scale=0.6]{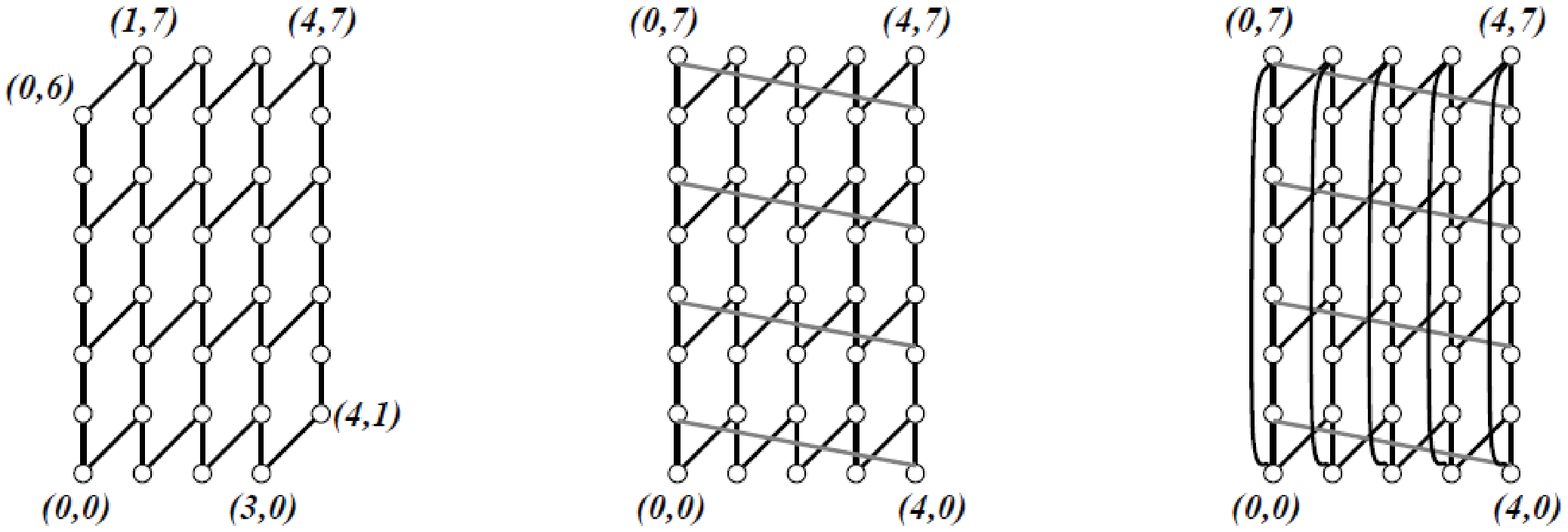}%\\~\\
 \begin{description}
    \item[Figure 9.] ${\rm PHG}(5,8)$ (left), ${\rm CHG}(5,8)$
    (middle), and ${\rm THG}(5,8)$ (right).
 \end{description}

 In Theorem 3.3 of \cite{Adams-Hexagonal-Grids2011},
 Adams et al. showed that if $G$ is an $m$ by $n$ planar hexagonal
 grid then
 {\rm min-seed}$(G,\theta_{\geq})=
 \left\lceil{(n-2)(m-1)\over 4}\right\rceil$.
 Below we consider an $m$ by $n$ planar hexagonal
 grid equipped with a strict majority threshold $\theta_>$ and
 determine its optimal target set.

  %
  %
  %  min-seed for (PHG(m,n),\theta_>)
  %
  \begin{theorem}
  \label{main-min-seed-PHG(m,n)}
  If $G$ is an $m$ by $n$ planar hexagonal
  grid, then
  {\rm min-seed}$(G,\theta_>)=
  \left\lceil{(mn+2m+n)/ 4}\right\rceil -1$.
  \end{theorem}
  \begin{proof}
  Let $G={\rm PHG}(m,n)$, $\theta_{\min}=\min\{\theta_>(v):v\in
  V(G)\}$.
  Obviously, {\rm min-seed}$(G,\theta_>)={\mbox{ min-seed}}(G,2)$.
  Since $G$ has $mn-2$ vertices and
  ${3mn\over 2}-{n\over 2}-m-2$ edges
  (see Lemma 3.1 of \cite{Adams-Hexagonal-Grids2011}), by
  the result min-seed$(G,\theta_>)\geq
  {|V(G)|\theta_{\min}-|E(G)|\over \theta_{\min}}$
  presented in Lemma \ref{lowerbound-for-target-set}, we see at once that
  {\rm min-seed}$(G,2)\geq \left\lceil{(mn+2m+n)/ 4}\right\rceil -1$.

  Next we will show that
  {\rm min-seed}$(G,2)\leq \left\lceil{(mn+2m+n)/ 4}\right\rceil -1$
  by giving a target set $S$ for $(G,2)$ which influences all
  vertices of $V(G)\setminus S$ and has
  $|S|=\left\lceil{(mn+2m+n)/ 4}\right\rceil
  -1$.
  %%%%%%%%%%%%%%%%%%%%%%%%%%%%%%%%%%%%%%%%%%%%%%%%% 2012-02-03 Yeh
  Note that $n$ is even and $n\geq 4$.
  Let $n=4t+r$, where $t$ is a positive integer and $r\in \{0,2\}$.
  The proof is divided into three cases, according to the value of
  $r$ and the parity of $m$.

  {\bf Case 1.} $r=2$. In this case, consider
   $S=\{(0,2i)|0\leq i\leq 2t\} \cup \{(j,4+4k)|1\leq j\leq m-1,0\leq k\leq t-1\}
   \cup \{(1,0),(2,0),(3,0), \ldots, (m-2,0)\} \cup \{(m-1,1)\}$
   as a target set for $(G,2)$ (see Figure 10 for a graphical illustration of $S$).
   It is straightforward to check that $S$ can influence all vertices of
   $V(G)\setminus S$ in $(G,2)$ by using the convinced sequence
   $\alpha=\alpha_1 \sqcup \alpha_2 \sqcup \alpha_3 \sqcup \alpha_4$
   (see Figure 5 in Appendix for a graphical illustration of this
   convinced sequence $\alpha$), where\\
   $\alpha_1=[(0,1),(1,1),(2,1),\ldots,(m-2,1)]$,\\
   $\alpha_2=\sqcup_{k=0}^{t-2}
   [(0,5+4k),(1,5+4k),(2,5+4k),\ldots,(m-1,5+4k)]$,\\
   $\alpha_3=\sqcup_{k=0}^{t-1} ([(0,3+4k),(1,3+4k),(1,2+4k)] \sqcup
   (\sqcup_{j=2}^{m-1} [(j,3+4k),(j,2+4k)]))$, and\\
   $\alpha_4=[(1,n-1),(2,n-1),(3,n-1),\ldots,(m-1,n-1)]$.\\
   Since
   $|S|=(2t+1)+(m-1)t+(m-1)=\lceil{(mn+2m+n)/4}\rceil-1
   $, we see that
   {\rm min-seed}$(G,2)\leq \lceil{(mn+2m+n)/4}\rceil-1$.

  {\bf Case 2.} $r=0$ and $m$ is even. Consider
   $S=\{(0,2i)|0\leq i\leq 2t-1\} \cup \{(j,2+4k)|1\leq j\leq m-1,0\leq k\leq t-1\}
   \cup \{(2,0),(4,0),(6,0), \ldots, (m-2,0)\}$
   as a target set for $(G,2)$ (see Figure 11 for a graphical illustration of $S$).
   It is straightforward to check that $S$ can influence all vertices of
   $V(G)\setminus S$ in $(G,2)$ by using the convinced sequence
   $\alpha=\alpha_1 \sqcup \alpha_2 \sqcup \alpha_3 \sqcup \alpha_4 \sqcup \alpha_5$
   (see Figure 6 in Appendix for a graphical illustration of this
   convinced sequence $\alpha$), where\\
   $\alpha_1=[(0,1),(1,1),(2,1),\ldots,(m-1,1)]$,\\
   $\alpha_2=[(1,0),(3,0),(5,0),\ldots,(m-3,0)]$,\\
   $\alpha_3=\sqcup_{k=0}^{t-2}
   [(0,3+4k),(1,3+4k),(2,3+4k),\ldots,(m-1,3+4k)]$,\\
   $\alpha_4=\sqcup_{k=0}^{t-2} ([(0,5+4k),(1,5+4k),(1,4+4k)] \sqcup
   (\sqcup_{j=2}^{m-1} [(j,5+4k),(j,4+4k)]))$, and\\
   $\alpha_5=[(1,n-1),(2,n-1),(3,n-1),\ldots,(m-1,n-1)]$.\\
   Since
   $|S|=2t+(m-1)t+(\frac{m}{2}-1)=\lceil{(mn+2m+n)/4}\rceil-1
   $, we see that
   {\rm min-seed}$(G,2)\leq \lceil{(mn+2m+n)/4}\rceil-1$.

  {\bf Case 3.} $r=0$ and $m$ is odd. Consider
   $S=\{(0,2i)|0\leq i\leq 2t-1\} \cup \{(j,2+4k)|1\leq j\leq m-1,0\leq k\leq t-1\}
   \cup \{(2,0),(4,0),(6,0), \ldots, (m-3,0)\} \cup \{(m-2,0)\}$
   as a target set for $(G,2)$ (see Figure 12 for a graphical illustration of $S$).
   It is straightforward to check that $S$ can influence all vertices of
   $V(G)\setminus S$ by using the convinced sequence
   $\alpha=\alpha_1 \sqcup \alpha_2 \sqcup \alpha_3 \sqcup \alpha_4 \sqcup \alpha_5$
   (see Figure 7 in Appendix for a graphical illustration of this
   convinced sequence $\alpha$), where\\
   $\alpha_1=[(0,1),(1,1),(2,1),\ldots,(m-1,1)]$,\\
   $\alpha_2=[(1,0),(3,0),(5,0),\ldots,(m-4,0)]$,\\
   $\alpha_3=\sqcup_{k=0}^{t-2}
   [(0,3+4k),(1,3+4k),(2,3+4k),\ldots,(m-1,3+4k)]$,\\
   $\alpha_4=\sqcup_{k=0}^{t-2} ([(0,5+4k),(1,5+4k),(1,4+4k)] \sqcup
   (\sqcup_{j=2}^{m-1} [(j,5+4k),(j,4+4k)]))$, and\\
   $\alpha_5=[(1,n-1),(2,n-1),(3,n-1),\ldots,(m-1,n-1)]$.\\
   Since
   $|S|=2t+(m-1)t+(\frac{m-1}{2})=\lceil{(mn+2m+n)/4}\rceil-1
   $, we see that
   {\rm min-seed}$(G,2)\leq \lceil{(mn+2m+n)/4}\rceil-1$.
  \end{proof}

 %%%%%%%%%%%%%%%%%%%%%%%%%%%%%
 %  Figure 10: seed-PHG_case1.eps   (n=4t+2)
 %%%%%%%%%%%%%%%%%%%%%%%%%%%%%
 \includegraphics[scale=0.6]{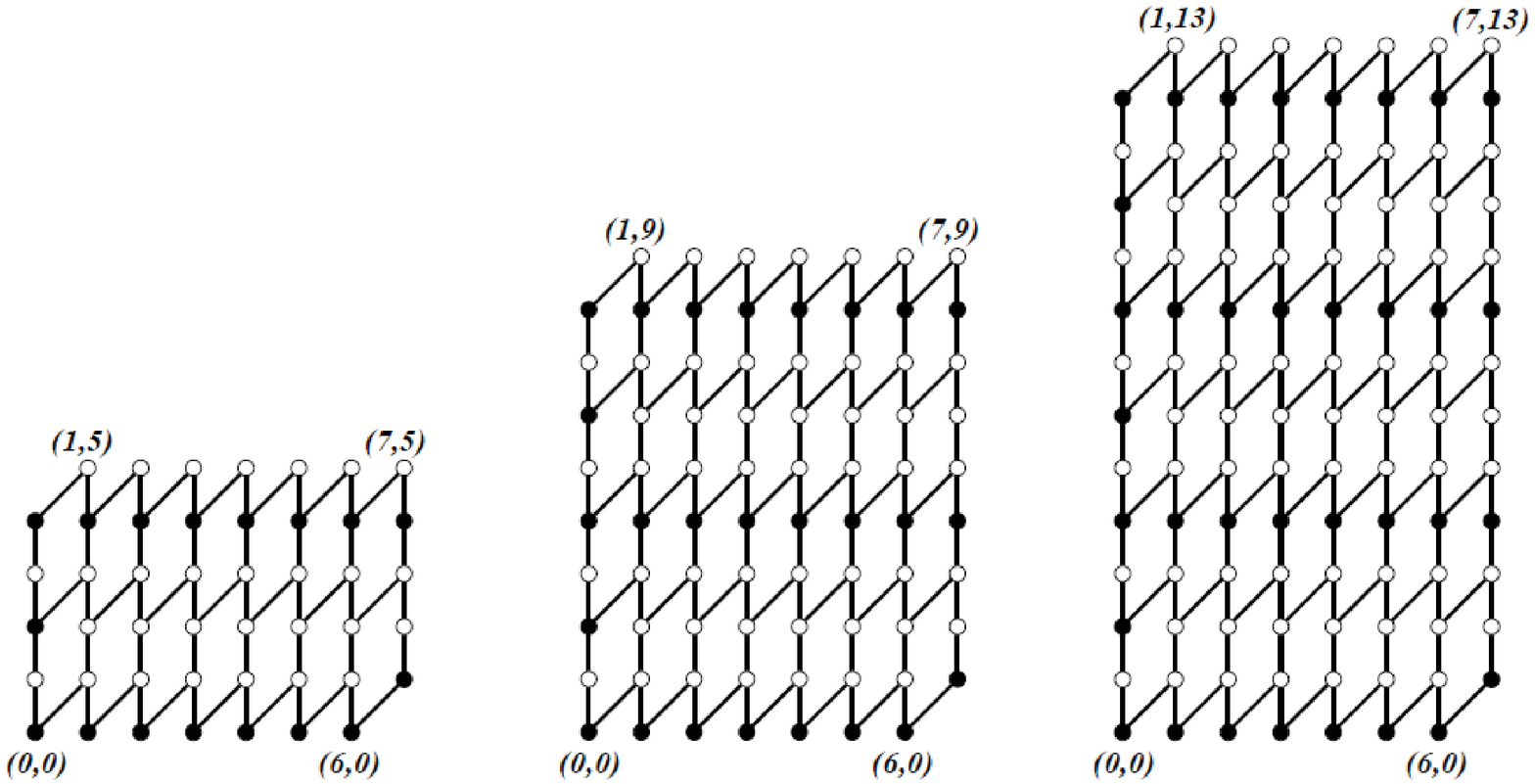}%\\~\\
 \begin{description}
    \item[Figure 10.] ${\rm PHG}(8,6)$ (left), ${\rm PHG}(8,10)$
    (middle), and ${\rm PHG}(8,14)$ (right)
    where the target set $S$ is the set of all black vertices.
 \end{description}

 %%%%%%%%%%%%%%%%%%%%%%%%%%%%%
 %  Figure 11: seed-PHG_case2.eps   (n=4t; m:even)
 %%%%%%%%%%%%%%%%%%%%%%%%%%%%%
 \includegraphics[scale=0.6]{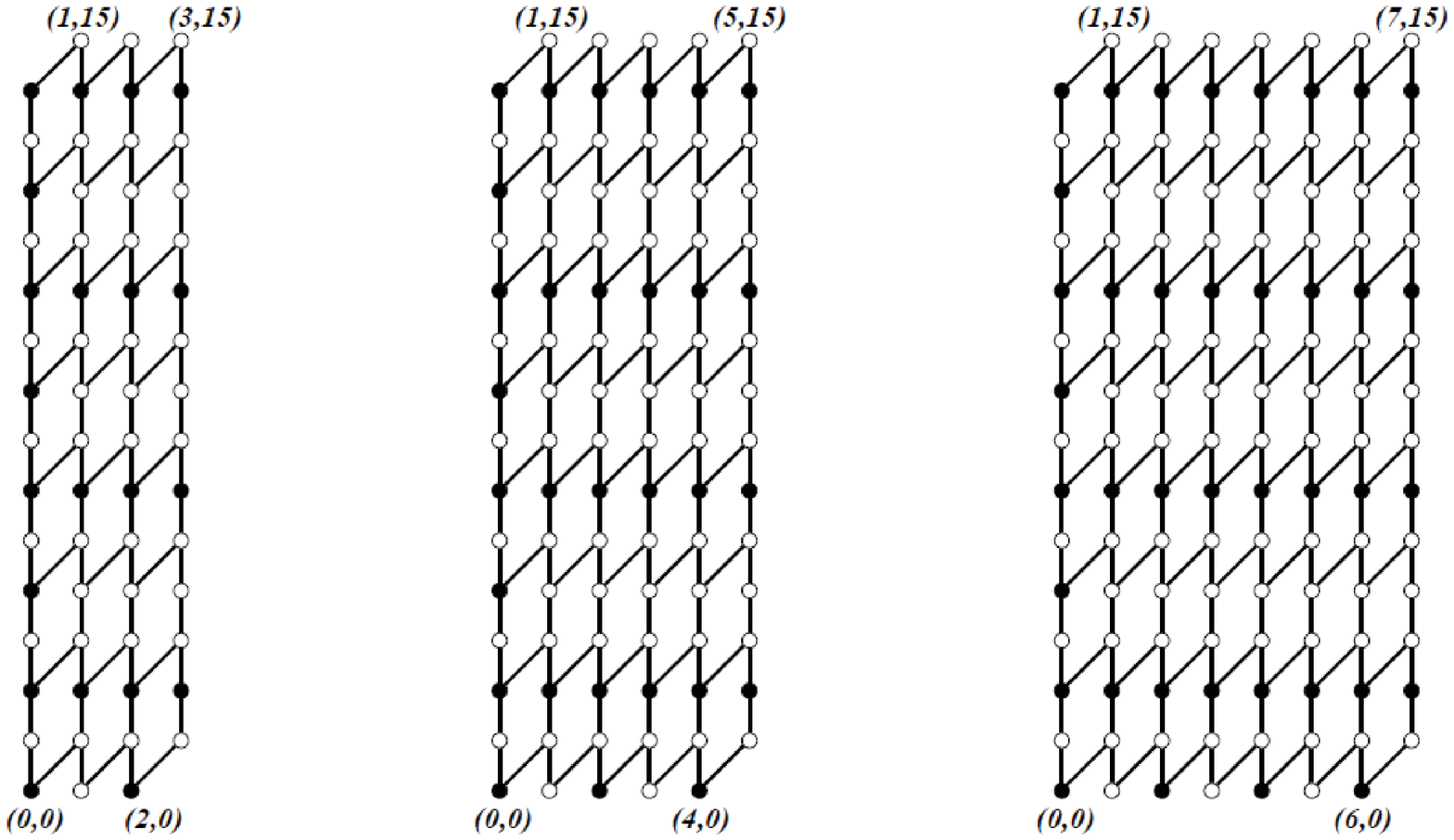}%\\~\\
 \begin{description}
    \item[Figure 11.] ${\rm PHG}(4,16)$ (left), ${\rm PHG}(6,16)$
    (middle), and ${\rm PHG}(8,16)$ (right)
    where the target set $S$ is the set of all black vertices.
 \end{description}

 %%%%%%%%%%%%%%%%%%%%%%%%%%%%%
 %  Figure 12: seed-PHG_case3.eps   (n=4t; m:odd)
 %%%%%%%%%%%%%%%%%%%%%%%%%%%%%
 \includegraphics[scale=0.6]{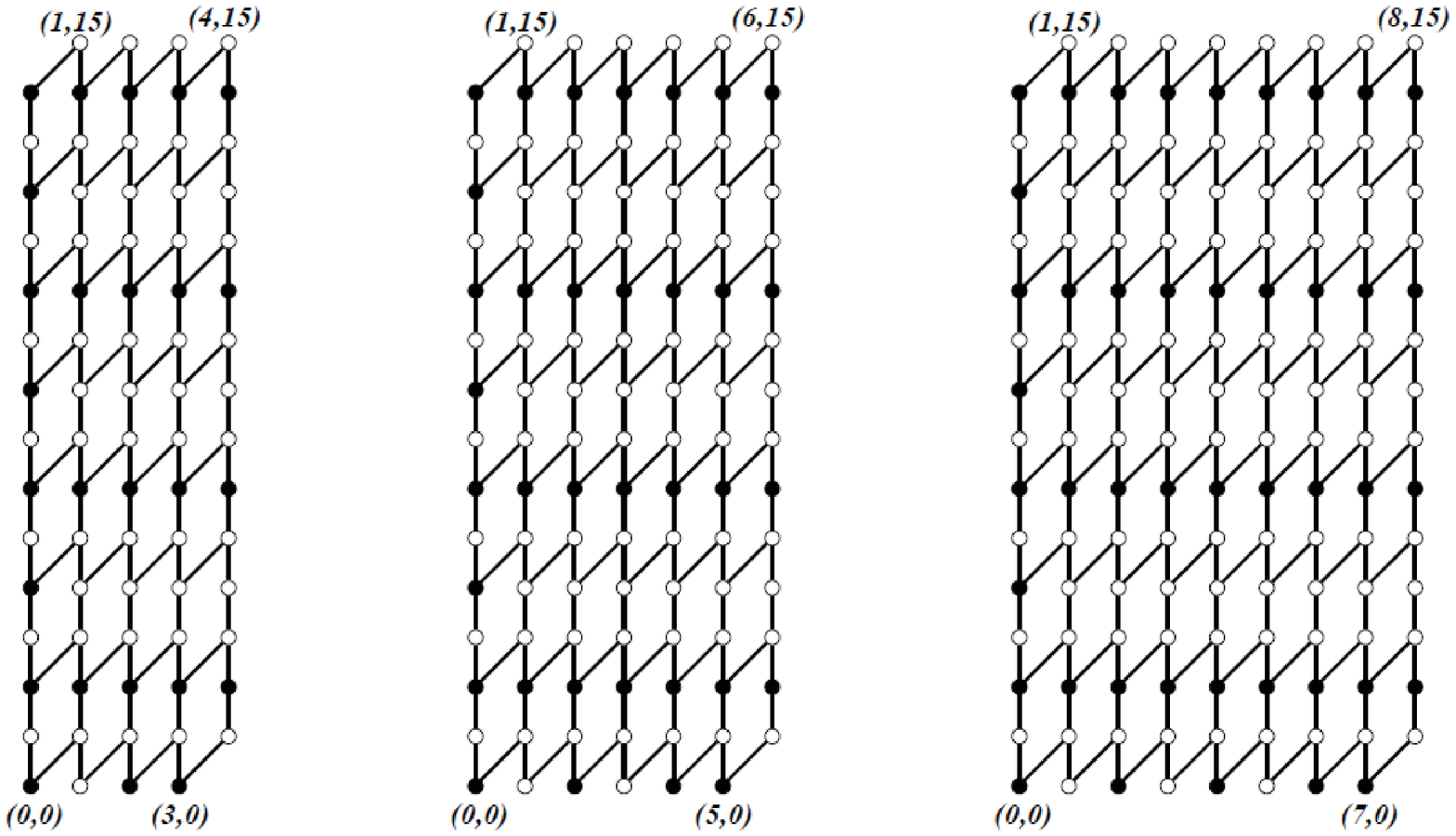}%\\~\\
 \begin{description}
    \item[Figure 12.] ${\rm PHG}(5,16)$ (left), ${\rm PHG}(7,16)$
    (middle), and ${\rm PHG}(9,16)$ (right)
    where the target set $S$ is the set of all black vertices.
 \end{description}

 In Theorems 4.1 and 4.2 of \cite{Adams-Hexagonal-Grids2011},
 Adams et al.~showed that if $G$ is an $m$ by $n$ cylindrical hexagonal
 grid, then
 {\rm min-seed}$(G,\theta_{\geq}) \in
 \{\lceil{(n-2)m+2\over 4}\rceil,\lceil{(n-2)m+2\over 4}\rceil+1\}$.
 Below we consider an $m$ by $n$ cylindrical hexagonal
 grid equipped with a strict majority threshold $\theta_>$ and
 determine its optimal target set.

  %
  %
  %  min-seed for (CHG(m,n),\theta_>)
  %
  \begin{theorem}
  \label{main-min-seed-CHG(m,n)}
  If $G$ is an $m$ by $n$ cylindrical hexagonal
  grid, then
  {\rm min-seed}$(G,\theta_>)=
  \left\lceil{(mn+2m)/ 4}\right\rceil$.
  \end{theorem}
  \begin{proof}
  Let $G={\rm CHG}(m,n)$, $\theta_{\min}=\min\{\theta_>(v):v\in
  V(G)\}$.
  Obviously, {\rm min-seed}$(G,\theta_>)={\mbox{ min-seed}}(G,2)$.
  Since $G$ has $mn$ vertices and ${3mn\over 2}-m$ edges
  (see Lemma 4.1 of \cite{Adams-Hexagonal-Grids2011}), by
  the result min-seed$(G,\theta_>)\geq
  {|V(G)|\theta_{\min}-|E(G)|\over \theta_{\min}}$
  presented in Lemma \ref{lowerbound-for-target-set}, we see at once that
  {\rm min-seed}$(G,2)\geq \left\lceil{(mn+2m)/ 4}\right\rceil$.

  Next we will show that
  {\rm min-seed}$(G,2)\leq \left\lceil{(mn+2m)/ 4}\right\rceil$
  by giving a target set $S$ for $(G,2)$ which influences all
  vertices of $V(G)\setminus S$ and has
  $|S|=\left\lceil{(mn+2m)/ 4}\right\rceil$.
  Notice that $n$ is even. Let
  $n=4t+r$ where $t$ is a positive integer and $r\in \{0,2\}$.
  The proof is divided into three cases, according to the value of
  $r$ and the parity of $m$.

  {\bf Case 1.} $r=2$. Consider
   $S=\{(0,2i)|0\leq i\leq 2t\} \cup \{(j,4k)|1\leq j\leq m-2,0\leq k\leq t\}
   \cup \{(m-1,n-2)\}$
   as a target set for $(G,2)$ (see Figure 13 for a graphical illustration of $S$).
   It is straightforward to check that $S$ can influence all vertices of
   $V(G)\setminus S$ by using the convinced sequence
   $\alpha=\alpha_1 \sqcup \alpha_2 \sqcup \alpha_3 \sqcup \alpha_4$
   (see Figure 8 in Appendix for a graphical illustration of this
   convinced sequence $\alpha$), where\\
   $\alpha_1=\sqcup_{k=0}^{t-1}
   [(0,1+4k),(1,1+4k),(2,1+4k),\ldots,(m-2,1+4k)]$,\\
   $\alpha_2=\sqcup_{k=0}^{t-1} ([(0,3+4k),(1,3+4k),(1,2+4k)] \sqcup
   (\sqcup_{j=2}^{m-2} [(j,3+4k),(j,2+4k)]))$,\\
   $\alpha_3=[(0,n-1),(1,n-1),(2,n-1),\ldots,(m-1,n-1)]$, and\\
   $\alpha_4=[(m-1,n-3),(m-1,n-4),(m-1,n-5),\ldots,(m-1,0)]$.\\
   Since
   $|S|=(2t+1)+(m-2)(t+1)+1=\lceil{(mn+2m)/4}\rceil$, we see that
   {\rm min-seed}$(G,2)\leq \lceil{(mn+2m)/4}\rceil$.

  {\bf Case 2.} $r=0$ and $m$ is even. Consider
   $S=\{(0,2i)|0\leq i\leq 2t-1\} \cup \{(j,2+4k)|1\leq j\leq m-2,0\leq k\leq t-1\}
   \cup \{(2,0),(4,0),(6,0),\ldots,(m-2,0)\} \cup \{(m-1,n-2)\}$
   as a target set for $(G,2)$ (see Figure 14 for a graphical illustration of $S$).
%   Let $G'$ be the subgraph of $G$ induced by the vertices in the
%   last $n-2$ rows of $G$. Clearly, $G'$ is an $m$ by $n-2$ cylindrical
%   hexagonal grid and $S\cap V(G')$ is an optimal target set for
%   $(G',2)$ which is proved in Case 1.
   It is straightforward to check that $S$ can influence all vertices of
   $V(G)\setminus S$ by using the convinced sequence
   $\alpha=\alpha_1 \sqcup \alpha_2 \sqcup \alpha_3 \sqcup \alpha_4
           \sqcup \beta_1 \sqcup \beta_2$
   (see Figure 9 in Appendix for a graphical illustration of this
   convinced sequence $\alpha$), where\\
   $\alpha_1=\sqcup_{k=0}^{t-2}
   [(0,3+4k),(1,3+4k),(2,3+4k),\ldots,(m-2,3+4k)]$,\\
   $\alpha_2=\sqcup_{k=0}^{t-2} ([(0,5+4k),(1,5+4k),(1,4+4k)] \sqcup
   (\sqcup_{j=2}^{m-2} [(j,5+4k),(j,4+4k)]))$,\\
   $\alpha_3=[(0,n-1),(1,n-1),(2,n-1),\ldots,(m-1,n-1)]$,\\
   $\alpha_4=[(m-1,n-3),(m-1,n-4),(m-1,n-5),\ldots,(m-1,2)]$,\\
   $\beta_1=[(0,1),(1,1),(2,1),\ldots,(m-1,1)]$, and\\
   $\beta_2=[(1,0),(3,0),(5,0),\ldots,(m-1,0)]$.\\
   Since
   $|S|=2t+(m-2)t+{m-2\over 2}+1=\lceil{(mn+2m)/4}\rceil$, we see that
   {\rm min-seed}$(G,2)\leq \lceil{(mn+2m)/4}\rceil$.

  {\bf Case 3.} $r=0$ and $m$ is odd. Consider
   $S=\{(0,2i)|0\leq i\leq 2t-1\} \cup \{(j,2+4k)|1\leq j\leq m-2,0\leq k\leq t-1\}
   \cup \{(2,0),(4,0),(6,0),\ldots,(m-3,0)\} \cup \{(m-1,1),(m-1,n-2)\}$
   as a target set for $(G,2)$ (see Figure 15 for a graphical illustration of $S$).
   It is straightforward to check that $S$ can influence all vertices of
   $V(G)\setminus S$ by using the convinced sequence
   $\alpha=\alpha_1 \sqcup \alpha_2 \sqcup \alpha_3 \sqcup \alpha_4
           \sqcup \beta_1 \sqcup \beta_2 \sqcup \beta_3$
   (see Figure 10 in Appendix for a graphical illustration of this
   convinced sequence $\alpha$), where\\
   $\alpha_1=\sqcup_{k=0}^{t-2}
   [(0,3+4k),(1,3+4k),(2,3+4k),\ldots,(m-2,3+4k)]$,\\
   $\alpha_2=\sqcup_{k=0}^{t-2} ([(0,5+4k),(1,5+4k),(1,4+4k)] \sqcup
   (\sqcup_{j=2}^{m-2} [(j,5+4k),(j,4+4k)]))$,\\
   $\alpha_3=[(0,n-1),(1,n-1),(2,n-1),\ldots,(m-1,n-1)]$,\\
   $\alpha_4=[(m-1,n-3),(m-1,n-4),(m-1,n-5),\ldots,(m-1,2)]$,\\
   $\beta_1=[(0,1),(1,1),(2,1),\ldots,(m-2,1)]$,\\
   $\beta_2=[(1,0),(3,0),(5,0),\ldots,(m-4,0)]$, and\\
   $\beta_3=[(m-2,0),(m-1,0)]$.\\
   Since
   $|S|=2t+(m-2)t+{m-3\over 2}+2=\lceil{(mn+2m)/4}\rceil$, we see that
   {\rm min-seed}$(G,2)\leq \lceil{(mn+2m)/4}\rceil$.
  \end{proof}

 %%%%%%%%%%%%%%%%%%%%%%%%%%%%%
 %  Figure 13: seed-CHG_case1.eps   (n=4t+2)
 %%%%%%%%%%%%%%%%%%%%%%%%%%%%%
 \includegraphics[scale=0.6]{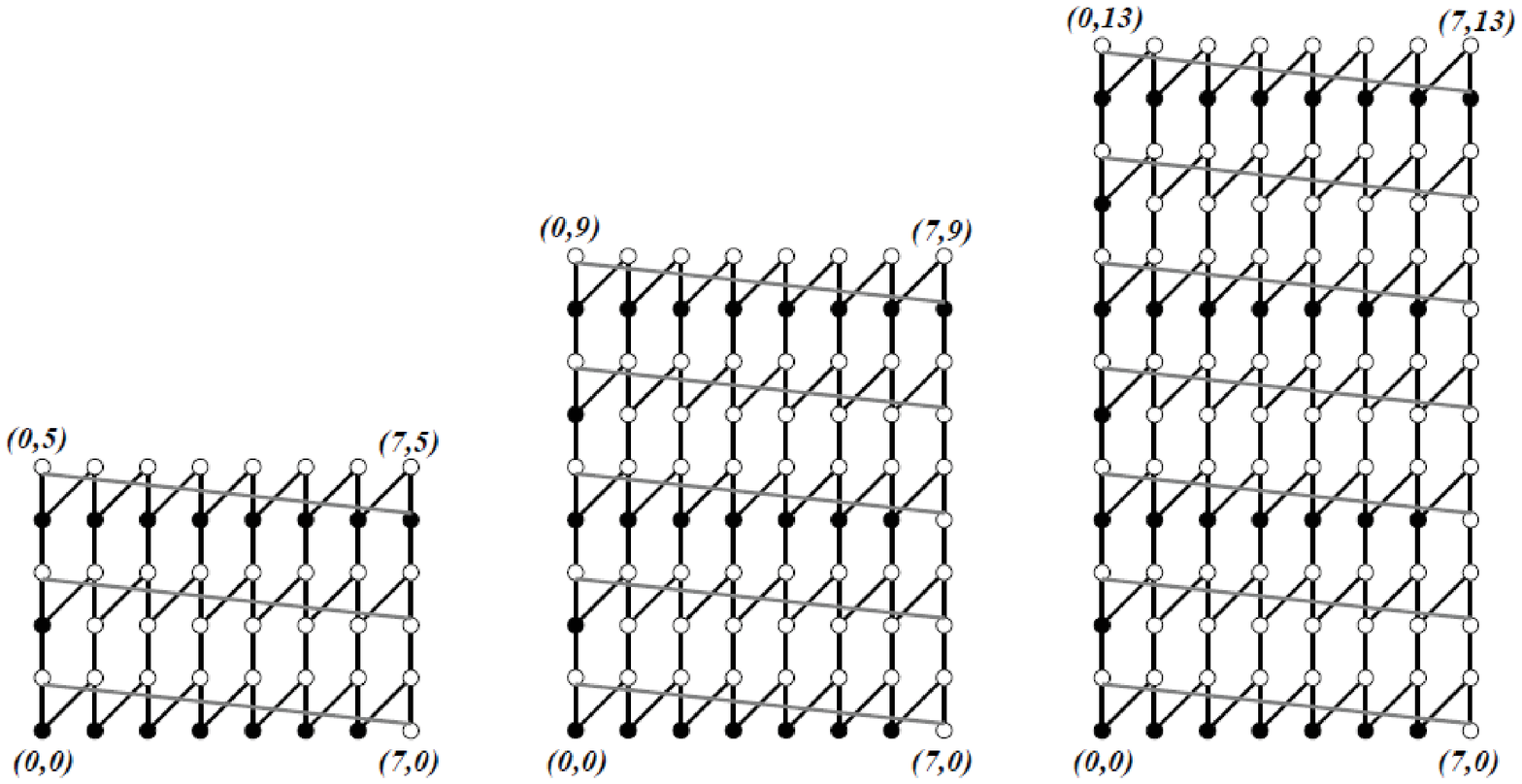}%\\~\\
 \begin{description}
    \item[Figure 13.] ${\rm CHG}(8,6)$ (left), ${\rm CHG}(8,10)$
    (middle), and ${\rm CHG}(8,14)$ (right)
    where the target set $S$ is the set of all black vertices.
 \end{description}

 %%%%%%%%%%%%%%%%%%%%%%%%%%%%%
 %  Figure 14: seed-CHG_case2.eps   (n=4t; m:even)
 %%%%%%%%%%%%%%%%%%%%%%%%%%%%%
 \includegraphics[scale=0.6]{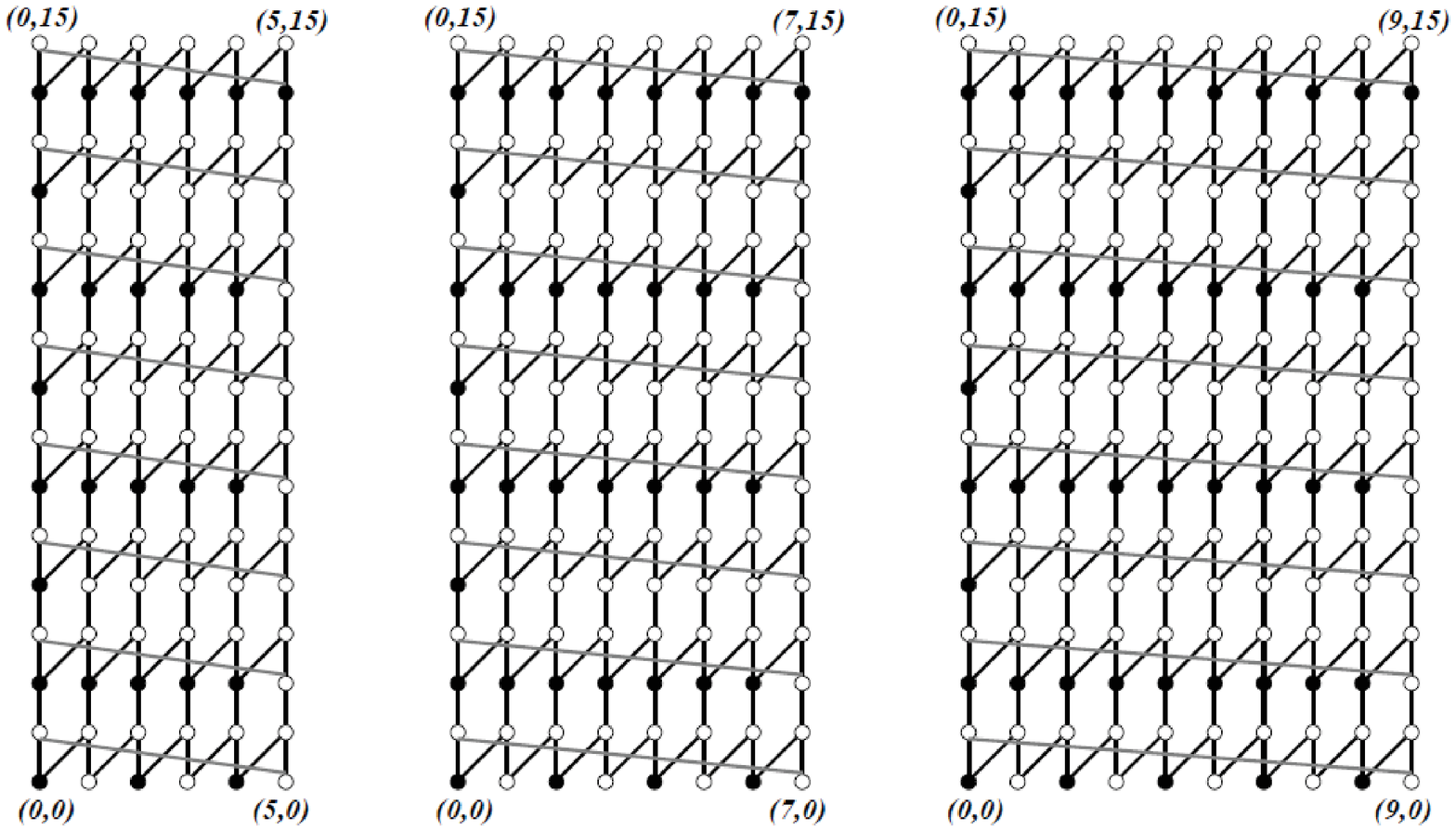}%\\~\\
 \begin{description}
    \item[Figure 14.] ${\rm CHG}(6,16)$ (left), ${\rm CHG}(8,16)$
    (middle), and ${\rm CHG}(10,16)$ (right)
    where the target set $S$ is the set of all black vertices.
 \end{description}

 %%%%%%%%%%%%%%%%%%%%%%%%%%%%%
 %  Figure 15: seed-CHG_case3.eps   (n=4t; m:odd)
 %%%%%%%%%%%%%%%%%%%%%%%%%%%%%
 \includegraphics[scale=0.6]{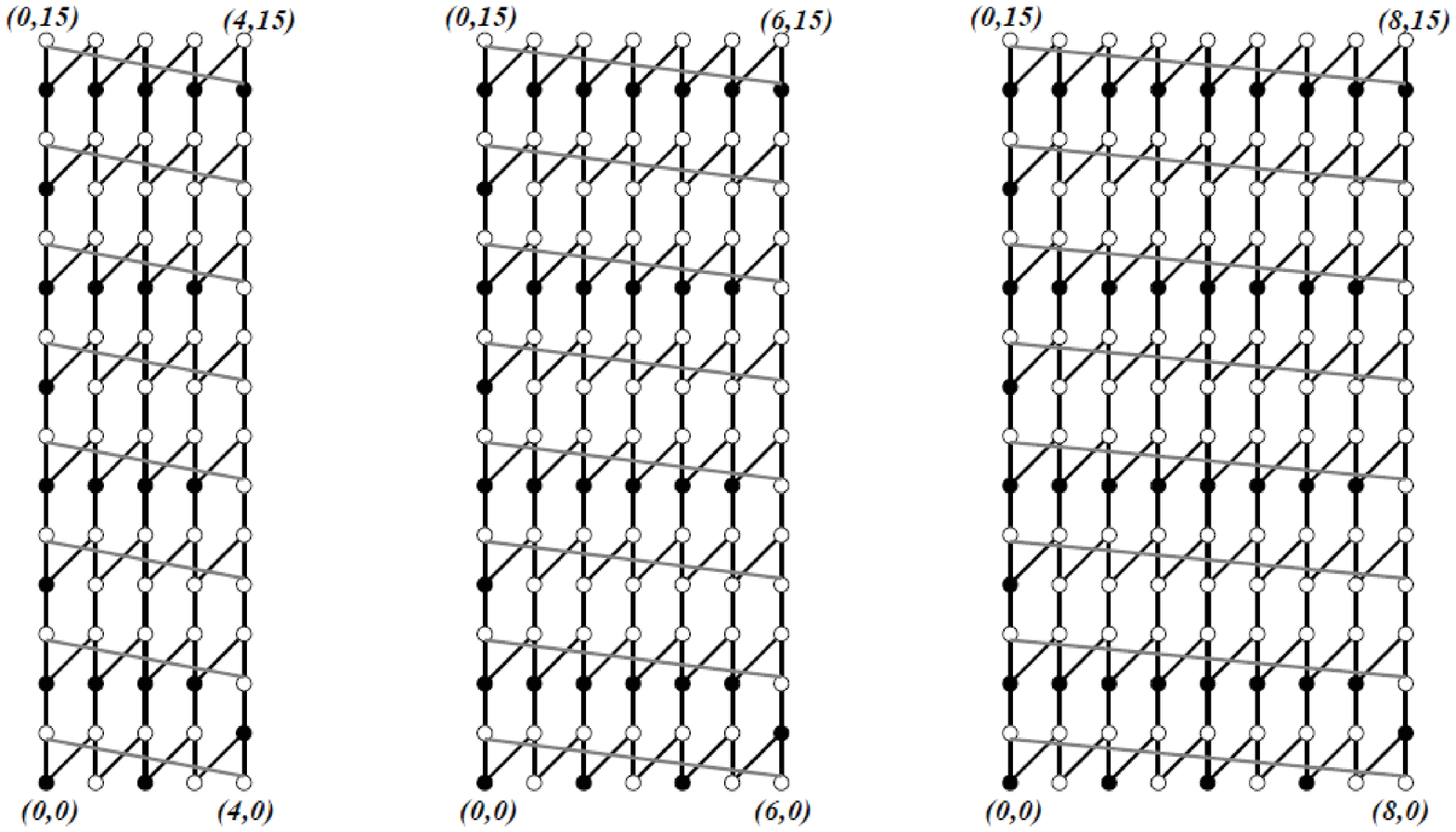}%\\~\\
 \begin{description}
    \item[Figure 15.] ${\rm CHG}(5,16)$ (left), ${\rm CHG}(7,16)$
    (middle), and ${\rm CHG}(9,16)$ (right)
    where the target set $S$ is the set of all black vertices.
 \end{description}

 In Theorems 5.1 and 5.2 of \cite{Adams-Hexagonal-Grids2011},
 Adams et al.~showed that if $G$ is an $m$ by $n$ toroidal hexagonal
 grid then
 {\rm min-seed}$(G,\theta_{\geq}) \in
 \{\lceil{mn+2\over 4}\rceil,\lceil{mn+2\over 4}\rceil+1\}$.
 Below we consider an $m$ by $n$ toroidal hexagonal
 grid equipped with a strict majority threshold $\theta_>$ and
 determine its optimal target set.
 Since ${\rm THG}(m,n)$ is $3$-regular,
 it can be seen that if $G$ is a toroidal hexagonal
 grid then {\rm min-seed}$(G,\theta_{\geq})=\mbox{
 min-seed}(G,\theta_{>})$.
 Thus
 our result in Theorem \ref{main-min-seed-THG(m,n)}
 closes the gap in the corresponding result proved by
 Adams et al.~(see Table 2).

  %
  %
  %  min-seed for (THG(m,n),\theta_>)
  %
  \begin{theorem}
  \label{main-min-seed-THG(m,n)}
  If $G$ is an $m$ by $n$ toroidal hexagonal
  grid, then
  {\rm min-seed}$(G,\theta_>)=\mbox{{\rm
  min-seed}$(G,\theta_\geq)$}=\nabla(G)=
  \lceil{(mn+2)/4}\rceil$.
  \end{theorem}
  \begin{proof}
  Let $G={\rm THG}(m,n)$.
  To prove {\rm min-seed}$(G,\theta_>)=
  \left\lceil{(mn+2)/4}\right\rceil$ we show that
  $G$ is isomorphic to
  the honeycomb rhombic torus ${\rm HRoT}(m,n)$.
  Let $f$ be a function from the vertex set of ${\rm HRoT}(m,n)$
  to the vertex set of $G$ such that
  $f(i,j)=(m-1-i,j-i)$.
  It is straightforward to check that
  $f$ is a bijection and preserves edges.
  Since both ${\rm HRoT}(m,n) $ and $G$
  have  ${3mn\over 2}$ edges, $f$ also preserves non-edges.
  Therefore $f$ is an isomorphism from ${\rm HRoT}(m,n)$ to $G$.
  It follows that, by Proposition 1 of \cite{Dreyer+Roberts} and
  Corollary \ref{main-corollary},
  {\rm min-seed}$(G,\theta_>)=\nabla(G)=
  \lceil{(mn+2)/4}\rceil$.
  \end{proof}

 %%%%%%%%%%%%%%%%%%%%%%%%%%%%%%%%%%%%%%%%%%%%
 %\section*{Acknowledgements}
 %%%%%%%%%%%%%%%%%%%%%%%%%%%%%%%%%%%%%%%%%%%%
 %
 %The authors thank anonymous referee for very careful reading
 %and for many constructive comments which helped to considerably
 %improve the presentation of the paper.
 %
 % The authors thank * for stimulating discussions, and an anonymous
 % referee for helpful commentary.

%%%%%%%%%%%%%%%%%%%%%%%%%%%%%%%%%%%%%%%%%%%%
%  Appendix
%%%%%%%%%%%%%%%%%%%%%%%%%%%%%%%%%%%%%%%%%%%%
\newpage

\noindent {\bf Appendix}:
[Not for publication - for referees'
information only]

 \begin{center}
 %%%%%%%%%%%%%%%%%%%%%%%%%%%%%
 %  Figure 1: Appendix01-cseq-HM4.eps
 %%%%%%%%%%%%%%%%%%%%%%%%%%%%%
 \includegraphics[scale=0.61]{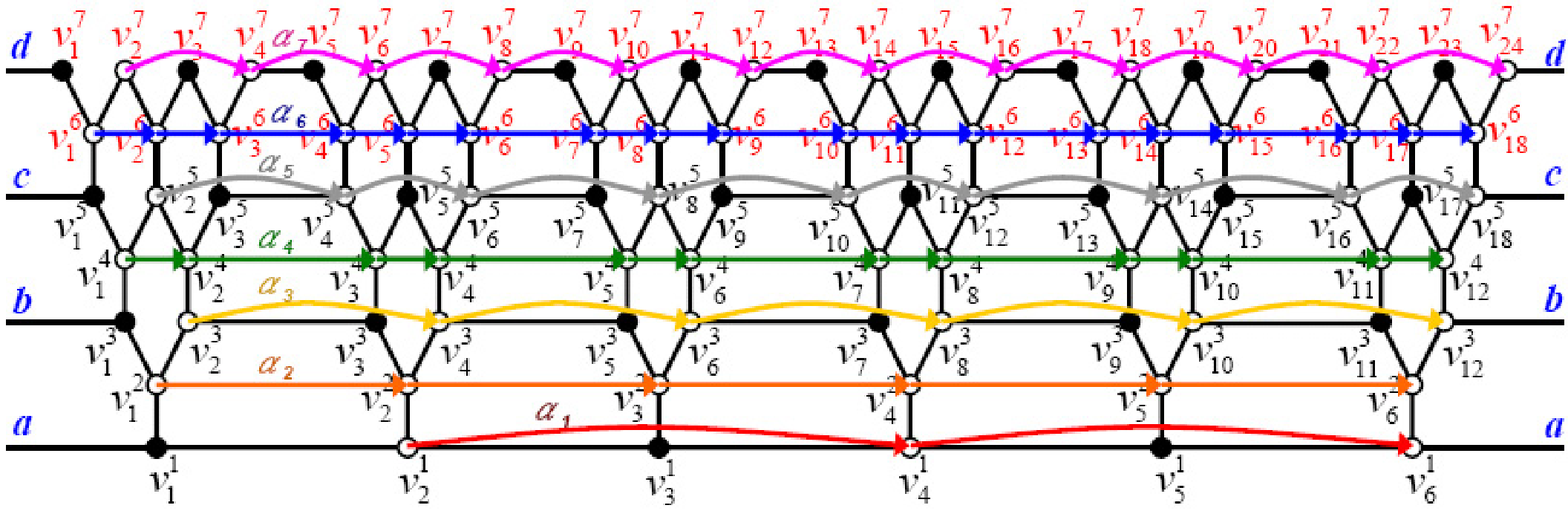}%\\~\\
 \begin{description}
    \item[Figure 1.]
    ${\rm HM}_4$ and its target set $S$. For $i=2,3,4$, convinced
    subsequences
    $\alpha_{2i-2}=[v_1^{2i-2},v_2^{2i-2},v_3^{2i-2},
        \ldots,v_{6(i-1)}^{2i-2}]$ and
        $\alpha_{2i-1}=[v_2^{2i-1},v_4^{2i-1},v_6^{2i-1},
        \ldots,v_{6i}^{2i-1}]$ are illustrated by colored
    directed paths.
    $\beta=\sqcup_{k=1}^{7}\alpha_k$.
 \end{description}
 \bigskip

 %%%%%%%%%%%%%%%%%%%%%%%%%%%%%

 %  Figure 2: Appendix02-cseq-GHT(6-8-4).eps

 %%%%%%%%%%%%%%%%%%%%%%%%%%%%%
 \includegraphics[scale=0.5]{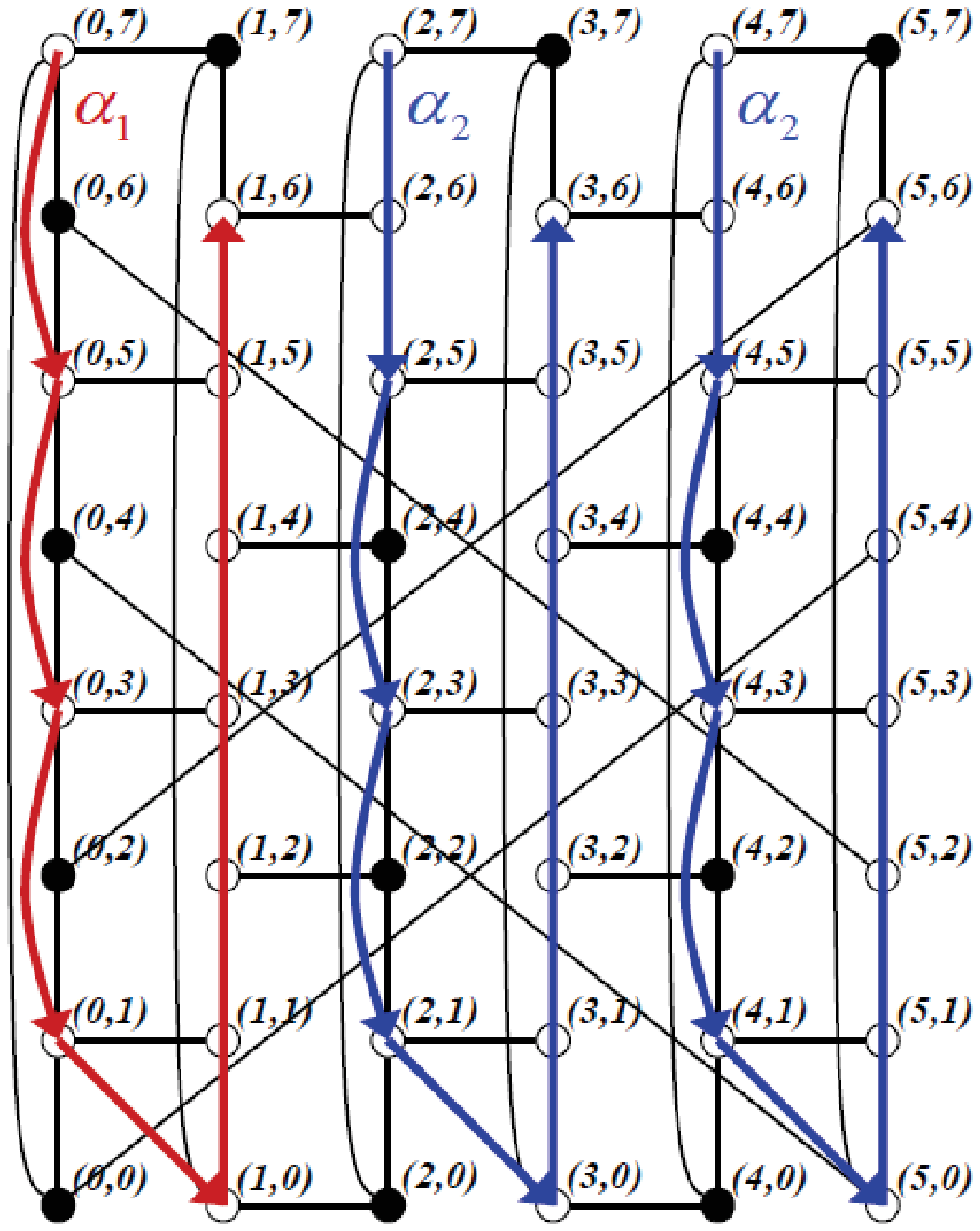}%\\~\\
 \begin{description}
    \item[Figure 2.]
    ${\rm GHT}(6,8,4)$ and its target sets $S$.
    Convinced subsequences $\alpha_1,\alpha_2$ are illustrated by colored
    directed paths and
    $\alpha=\alpha_1 \sqcup \alpha_2$.
 \end{description}
 \bigskip

 %%%%%%%%%%%%%%%%%%%%%%%%%%%%%

 %  Figure 3: Appendix03-cseq-GHT(9-8-5).eps

 %%%%%%%%%%%%%%%%%%%%%%%%%%%%%
 \includegraphics[scale=0.5]{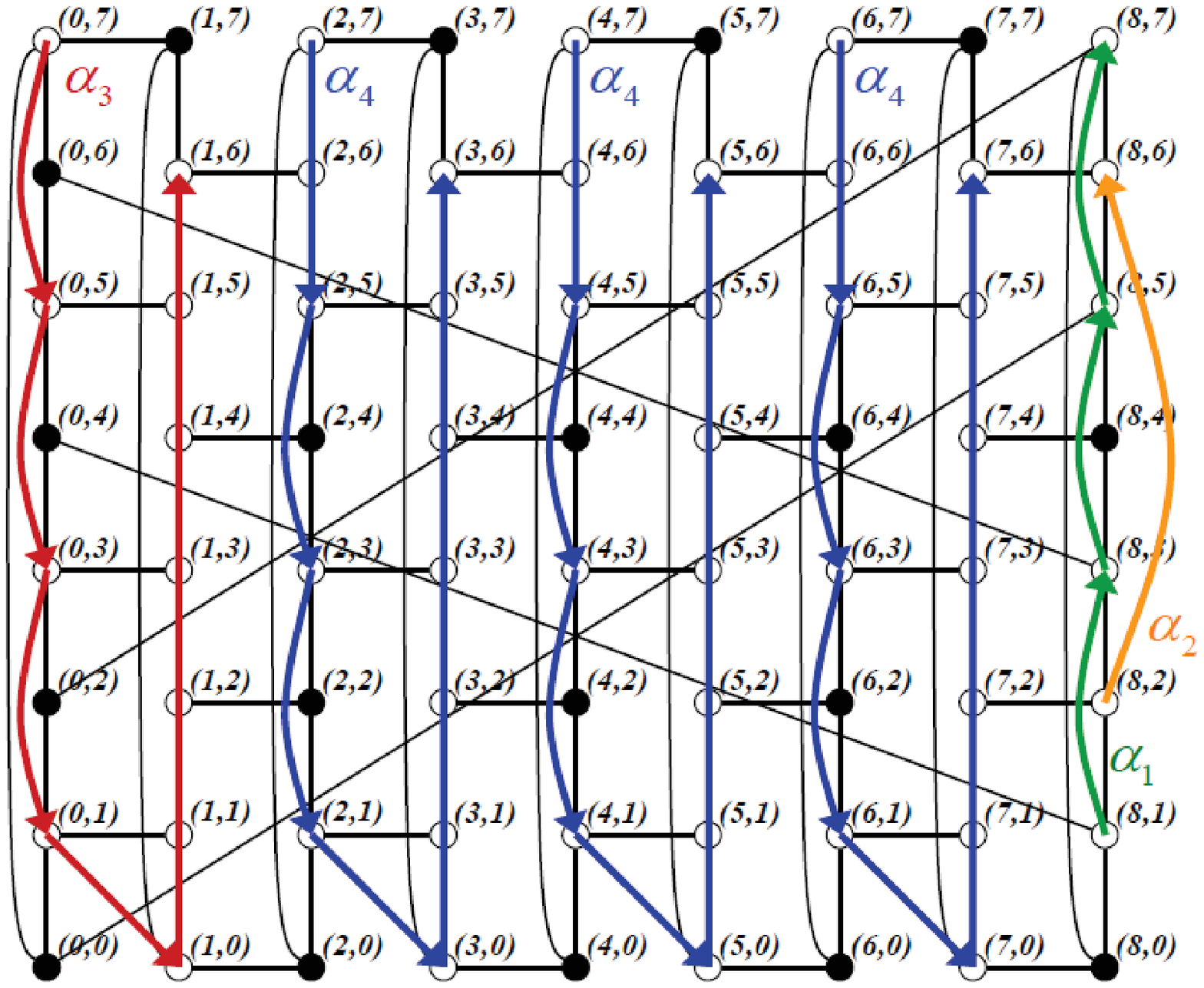}%\\~\\
 \begin{description}
    \item[Figure 3.]
    ${\rm GHT}(9,8,5)$ and its target sets $S$.
    Convinced subsequences $\alpha_1,\ldots,\alpha_4$ are illustrated by colored
    directed paths and
    $\alpha=\sqcup_{i=1}^{4} \alpha_i$.
 \end{description}
 \bigskip

 %%%%%%%%%%%%%%%%%%%%%%%%%%%%%

 %  Figure 4: Appendix04-cseq-GHT(9-10-5).eps

 %%%%%%%%%%%%%%%%%%%%%%%%%%%%%
 \includegraphics[scale=0.5]{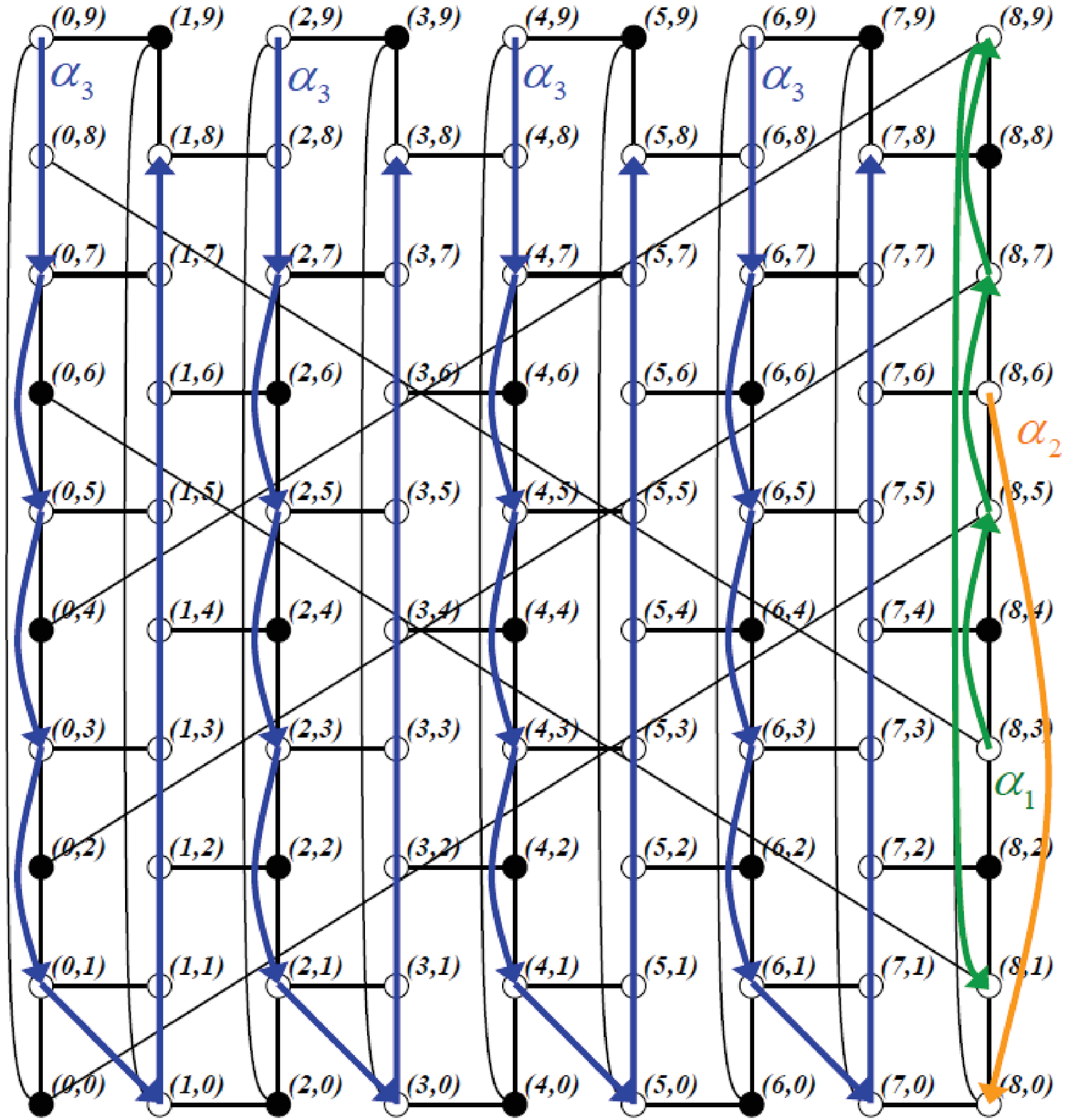}%\\~\\
 \begin{description}
    \item[Figure 4.]
    ${\rm GHT}(9,10,5)$ and its target sets $S$.
    Convinced subsequences $\alpha_1,\alpha_2,\alpha_3$ are illustrated by colored
    directed paths and
    $\alpha=\sqcup_{i=1}^{3} \alpha_i$.
 \end{description}
 \bigskip

 %%%%%%%%%%%%%%%%%%%%%%%%%%%%%

 %  Figure 5: Appendix05-cseq-PHG_case1.eps

 %%%%%%%%%%%%%%%%%%%%%%%%%%%%%
 \includegraphics[scale=0.6]{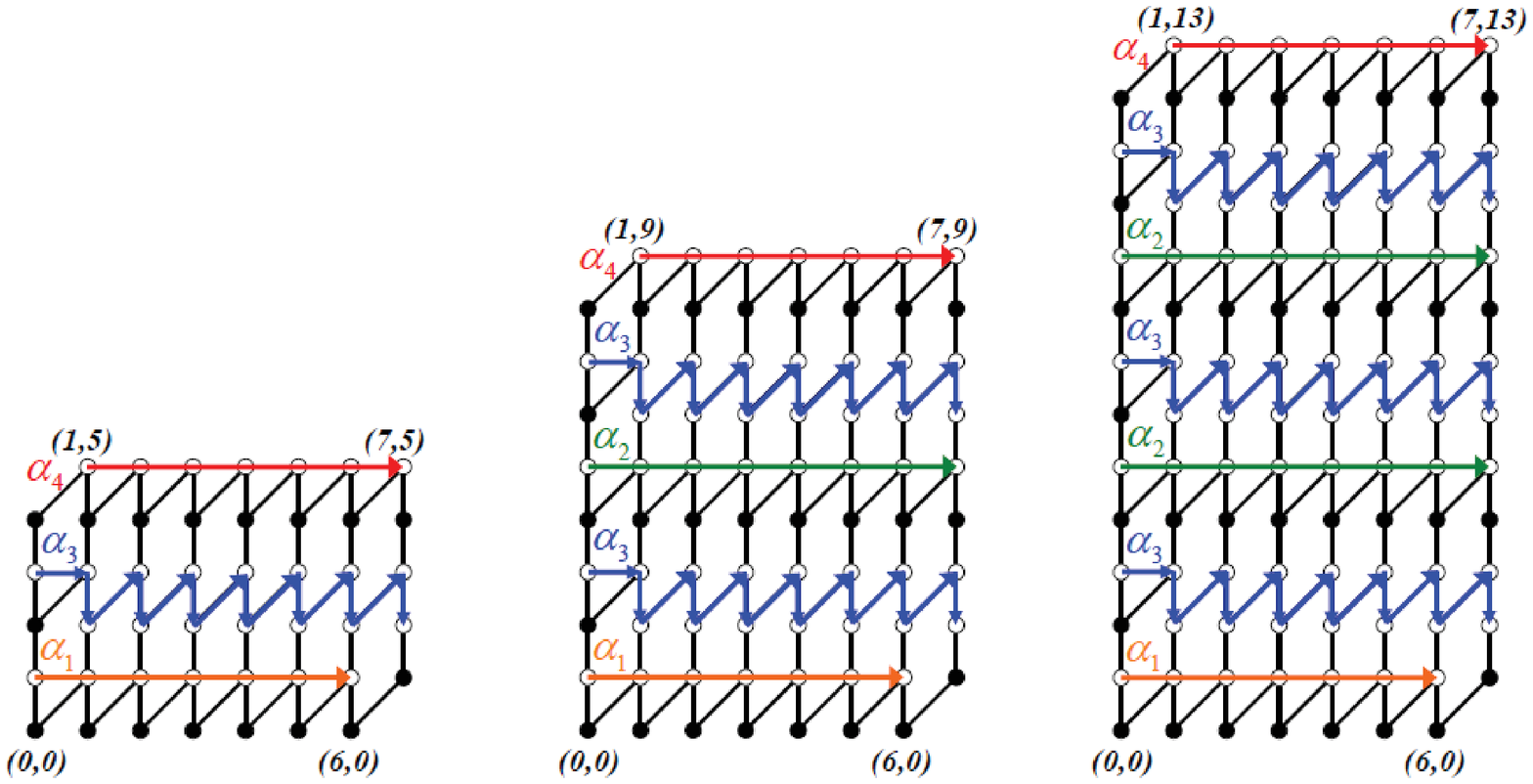}%\\~\\
 \begin{description}
    \item[Figure 5.]
    ${\rm PHG}(8,6)$, ${\rm PHG}(8,10)$, ${\rm PHG}(8,14)$ and their target sets $S$.
    Convinced subsequences $\alpha_1,\ldots,\alpha_4$ are illustrated by colored
    directed paths and
    $\alpha=\sqcup_{k=1}^{4}\alpha_k$.
 \end{description}
 \bigskip

 %%%%%%%%%%%%%%%%%%%%%%%%%%%%%

 %  Figure 6: Appendix06-cseq-PHG_case2.eps

 %%%%%%%%%%%%%%%%%%%%%%%%%%%%%
 \includegraphics[scale=0.6]{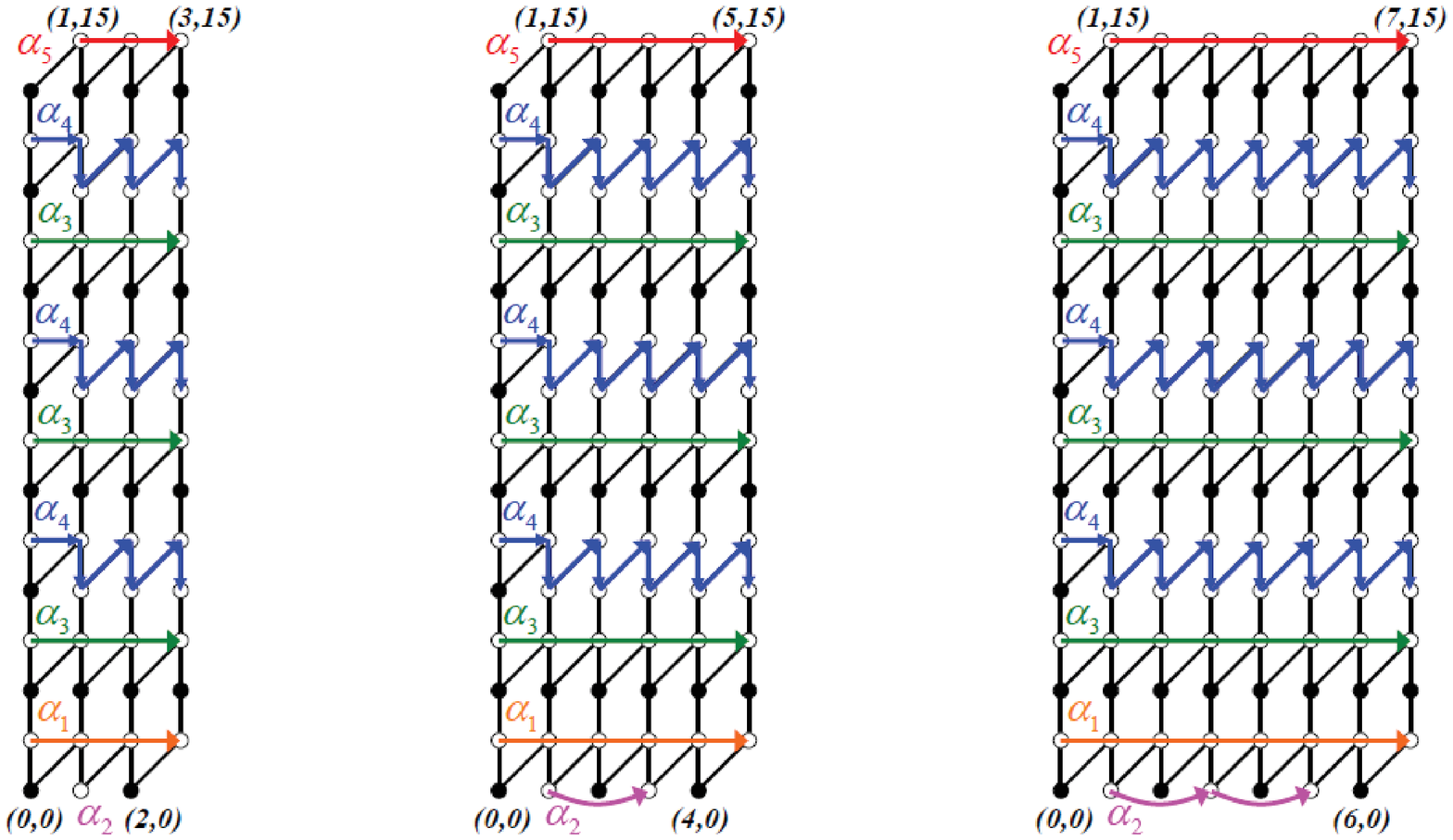}%\\~\\
 \begin{description}
    \item[Figure 6.]
    ${\rm PHG}(4,16)$, ${\rm PHG}(6,16)$, ${\rm PHG}(8,16)$ and their target sets $S$.
    Convinced subsequences $\alpha_1,\ldots,\alpha_5$ are illustrated by colored
    directed paths and
    $\alpha=\sqcup_{k=1}^{5}\alpha_k$.
 \end{description}
 \bigskip

 %%%%%%%%%%%%%%%%%%%%%%%%%%%%%

 %  Figure 7: Appendix07-cseq-PHG_case3.eps

 %%%%%%%%%%%%%%%%%%%%%%%%%%%%%
 \includegraphics[scale=0.6]{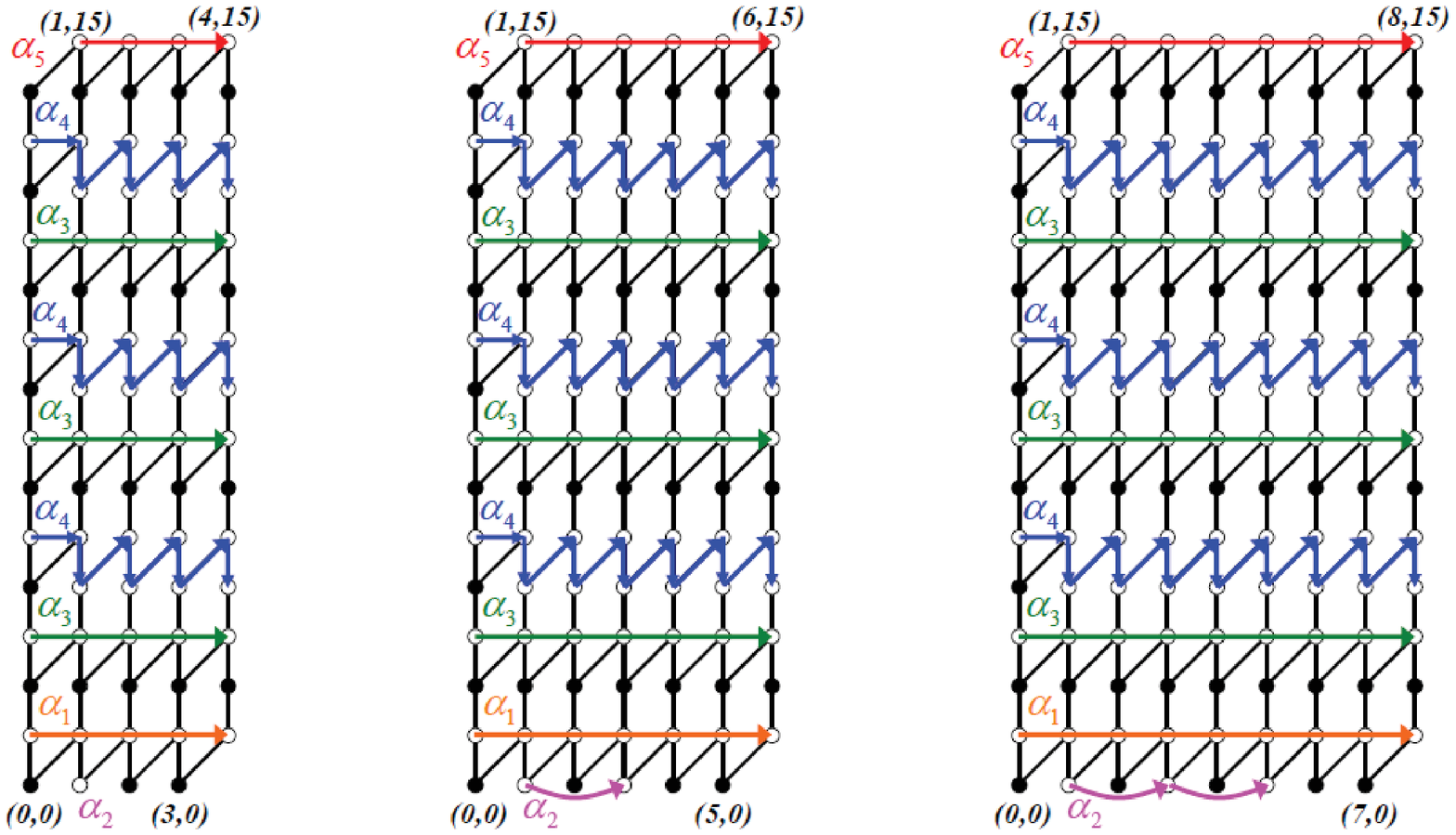}%\\~\\
 \begin{description}
    \item[Figure 7.]
    ${\rm PHG}(5,16)$, ${\rm PHG}(7,16)$, ${\rm PHG}(9,16)$ and their target sets $S$.
    Convinced subsequences $\alpha_1,\ldots,\alpha_5$ are illustrated by colored
    directed paths and
    $\alpha=\sqcup_{k=1}^{5}\alpha_k$.
 \end{description}
 \bigskip

 %%%%%%%%%%%%%%%%%%%%%%%%%%%%%

 %  Figure 8: Appendix08-cseq-CHG_case1.eps

 %%%%%%%%%%%%%%%%%%%%%%%%%%%%%
 \includegraphics[scale=0.6]{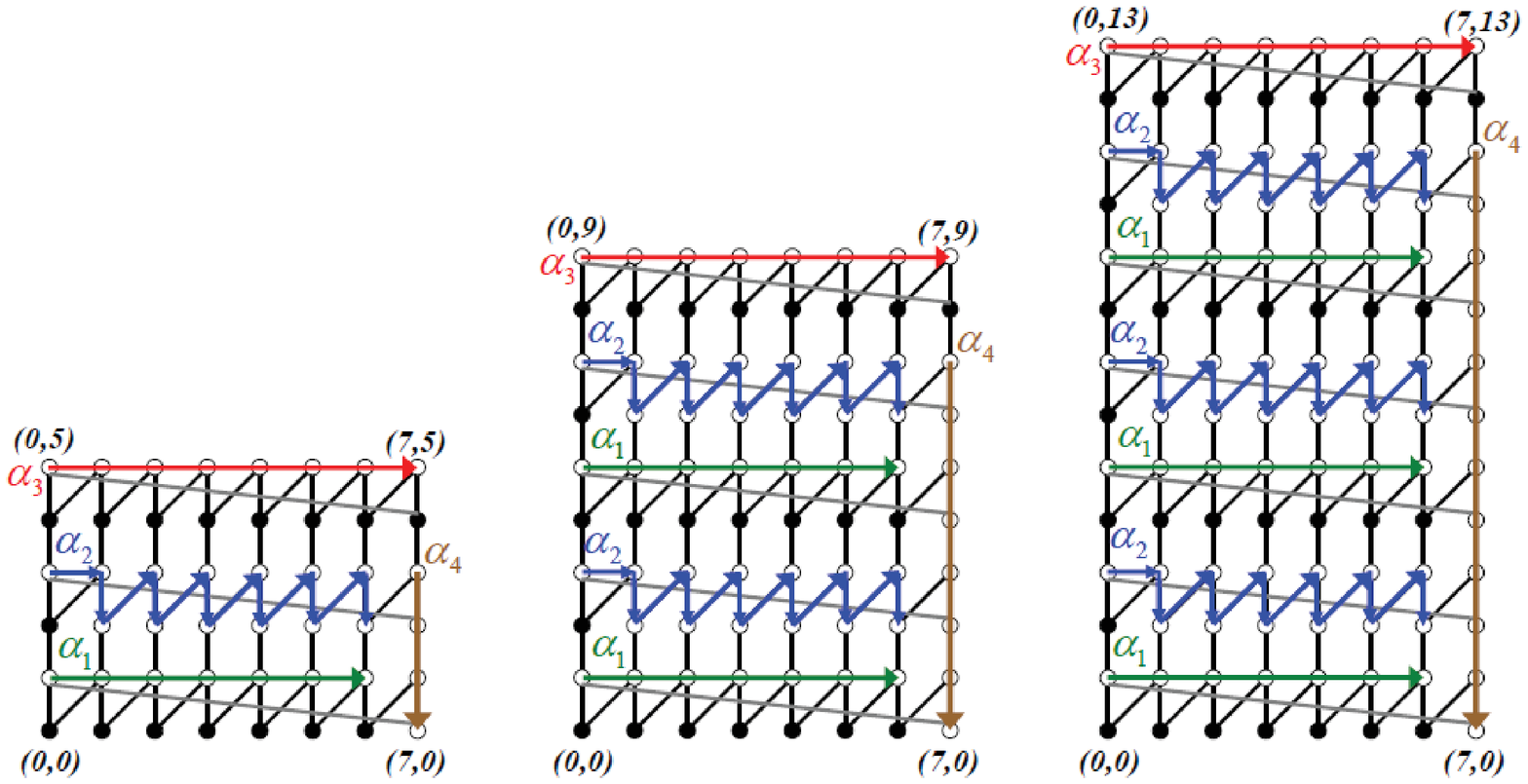}%\\~\\
 \begin{description}
    \item[Figure 8.]
    ${\rm CHG}(8,6)$, ${\rm CHG}(8,10)$, ${\rm CHG}(8,14)$ and their target sets $S$.
    Convinced subsequences $\alpha_1,\ldots,\alpha_4$ are illustrated by colored
    directed paths and
    $\alpha=\sqcup_{k=1}^{4}\alpha_k$.
 \end{description}
 \bigskip

 %%%%%%%%%%%%%%%%%%%%%%%%%%%%%

 %  Figure 9: Appendix09-cseq-CHG_case2.eps

 %%%%%%%%%%%%%%%%%%%%%%%%%%%%%
 \includegraphics[scale=0.59]{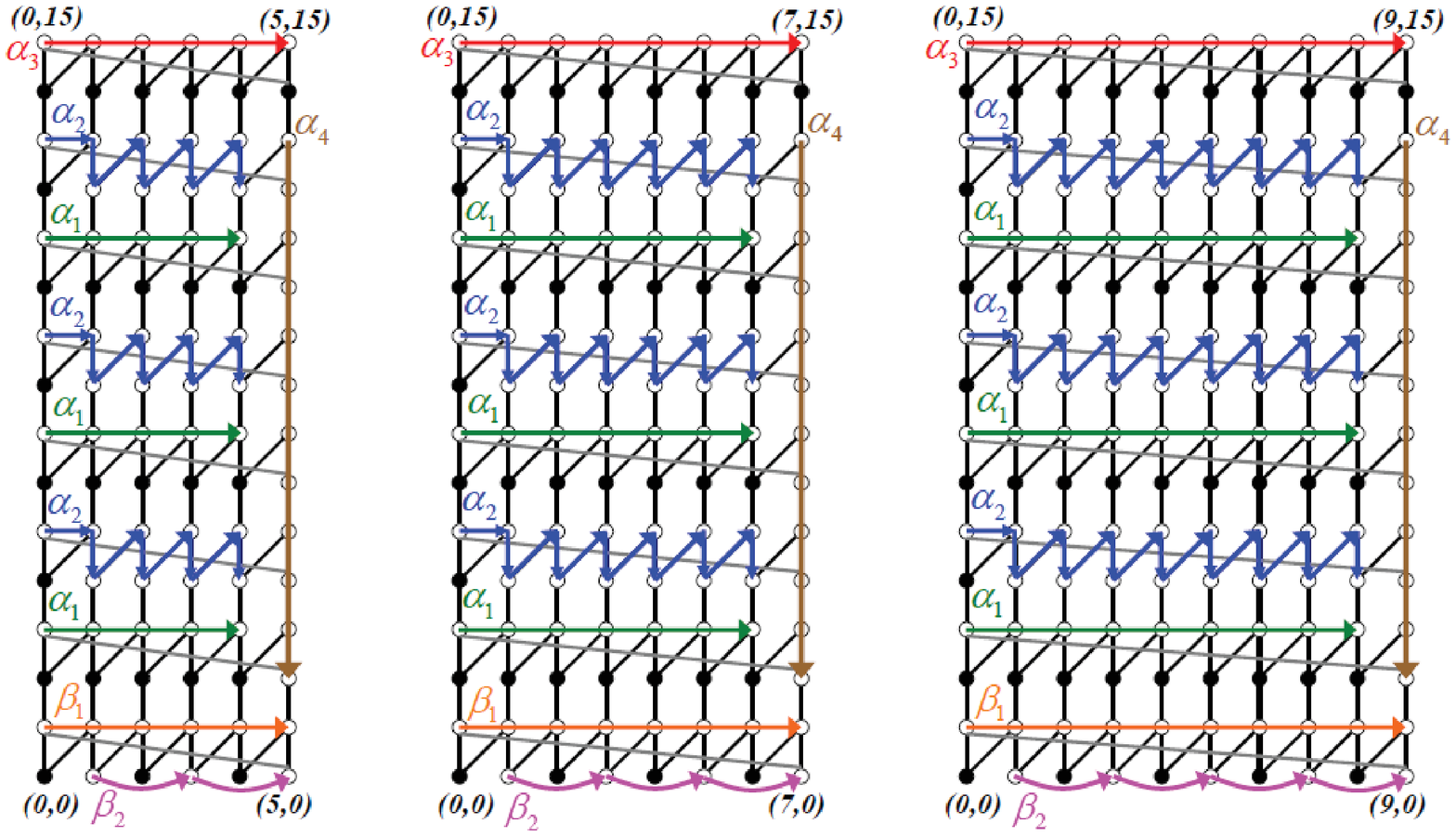}%\\~\\
 \begin{description}
    \item[Figure 9.]
    ${\rm CHG}(6,16)$, ${\rm CHG}(8,16)$, ${\rm CHG}(10,16)$ and their target sets $S$.
    Convinced subsequences $\alpha_1,\ldots,\alpha_4$ and $\beta_1,\beta_2$ are illustrated by colored
    directed paths and
    $\alpha=(\sqcup_{k=1}^{4}\alpha_k)\sqcup \beta_1 \sqcup \beta_2 $.
 \end{description}
 \bigskip

 %%%%%%%%%%%%%%%%%%%%%%%%%%%%%

 %  Figure 10: Appendix10-cseq-CHG_case3.eps

 %%%%%%%%%%%%%%%%%%%%%%%%%%%%%
 \includegraphics[scale=0.6]{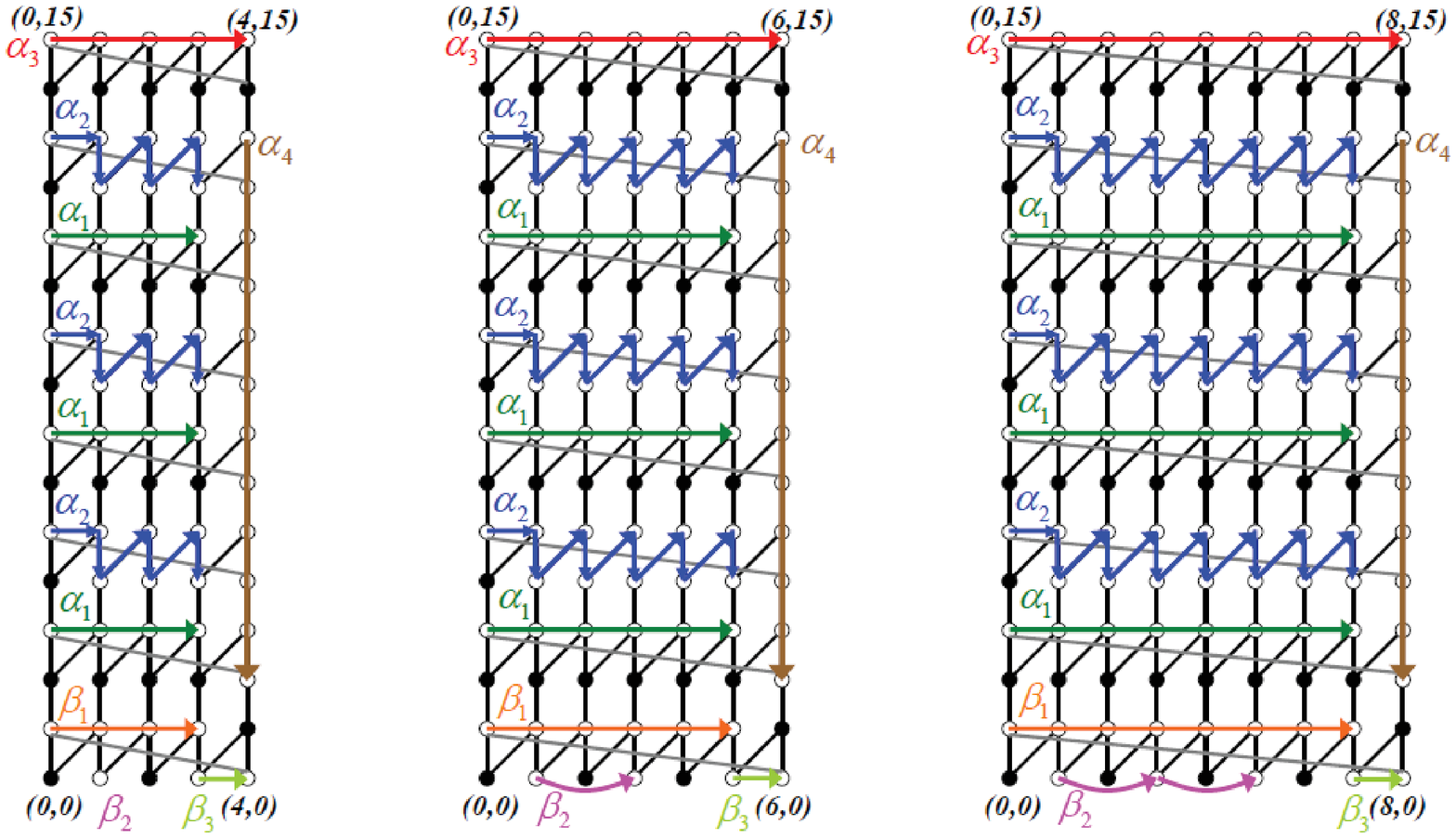}%\\~\\
 \begin{description}
    \item[Figure 10.]
    ${\rm CHG}(5,16)$, ${\rm CHG}(7,16)$, ${\rm CHG}(9,16)$ and their target sets $S$.
    Convinced subsequences $\alpha_1,\ldots,\alpha_4$ and $\beta_1,\beta_2,\beta_3$ are illustrated by colored
    directed paths and
    $\alpha=(\sqcup_{k=1}^{4}\alpha_k)\sqcup \beta_1 \sqcup \beta_2\sqcup \beta_3$.
 \end{description}
 \bigskip

\end{center}

\end{document}